\documentclass[11pt]{article}
\usepackage[margin=1in]{geometry}
\usepackage{amsthm,amsmath}
\usepackage{amsmath}
\usepackage{mathtools}
\usepackage{verbatim}

\usepackage[ruled]{algorithm}
\usepackage[noend]{algpseudocode}

\usepackage{enumerate}

\usepackage{amsmath,amsthm}
\usepackage{amssymb}
\usepackage{amsfonts}
\usepackage{paralist}
\usepackage[ruled]{algorithm}
\usepackage[noend]{algpseudocode}
\usepackage{hyperref}

\usepackage{graphicx}
\usepackage{caption}
\usepackage{subcaption}

\usepackage{epsfig}
\usepackage{enumerate}

\usepackage{chngcntr}
\usepackage{apptools}

\newcommand{\remove}[1]{}

\usepackage{color}

\newtheorem{definition}{Definition}
\newtheorem{theorem}{Theorem}
\newtheorem{corollary}{Corollary}
\newtheorem{lemma}{Lemma}

\newtheorem{fact}{Fact}

\newif\ifrobocza
\roboczatrue

\newif\ifbibtex
\bibtexfalse

\ifrobocza
\newcommand{\tj}[1]{{\color{blue}{#1}}}
\else
\newcommand{\todo}[1]{}
\newcommand{\tj}[1]{#1}
\fi

\newcommand{\cL}{{\mathcal L}}

\newcommand{\zet}{\text{lg } n}

\newcommand{\Set}[2]{%
	\{\, #1 \mid #2 \, \}%
}

\newcommand{\Tup}[2]{%
	(#1, #2) 
}

\DeclarePairedDelimiter\ceil{\lceil}{\rceil}
\DeclarePairedDelimiter\floor{\lfloor}{\rfloor}

\newcommand{\executor}{\textsc{Executor}}
\newcommand{\radionetwork}{{\sf Radio Network}}
\newcommand{\execack}{\textsc{ExecAck}}
\newcommand{\level}{\textit{level}}
\newcommand{\layer}{\textit{layer}}
\newcommand{\leaf}{\textit{leaf}}
\newcommand{\parent}{\textit{parent}}
\newcommand{\fastsd}{\textsc{FastSD}}

\begin{document}
	
\remove{
	\title{On Labelings for Network Topology Recognition in Radio Model
		\footnote{%
			This work was supported by the Polish National Science Centre 
			grant DEC-2017/25/B/ST6/02010.
		}
}
	

	\author[1]{Adam Ga\'{n}czorz}
	\author[1]{Tomasz Jurdzi\'{n}ski}
	\author[1]{Mateusz Lewko}
	\author[2]{Andrzej Pelc}

	\affil[1]{
	Institute of Computer Science, University of Wroc{\l}aw, Poland}
	\affil[2]{Departement dâinformatique, Universite du Quebec en Outaouais, Gatineau, Quebec, Canada}
	
	\date{}
	
}

\def\thefootnote{\fnsymbol{footnote}}

\title{{\bf Deterministic Size Discovery and Topology Recognition in Radio Networks with Short Labels}}

\author{Adam Ga\'{n}czorz\footnotemark[1]
	\and Tomasz Jurdzi\'{n}ski\footnotemark[1] 
	\and Mateusz Lewko\footnotemark[1] 
	\and  Andrzej Pelc\footnotemark[2]
}

\footnotetext[1]{Institute of Computer Science, University of Wroc{\l}aw, Poland. 
	Emails: {\tt \{adam.ganczorz,tju\}@cs.uni.wroc.pl, mateusz.lewko@gmail.com}. 
	Supported by the Polish National Science Centre 
	grant DEC-2017/25/B/ST6/02010.}

\footnotetext[2]{
	D\'epartement d'informatique, Universit\'e du Qu\'ebec en Outaouais, Gatineau,
	Qu\'ebec J8X 3X7, Canada. {\tt pelc@uqo.ca}. Partially supported by NSERC discovery grant 2018-03899
	and by the Research Chair in Distributed Computing at the
	Universit\'e du Qu\'ebec en Outaouais.}

\maketitle

\thispagestyle{empty}

\begin{abstract}
We consider the fundamental problems of \emph{size discovery} and \emph{topology recognition} in radio networks modeled by simple undirected connected graphs.
Size discovery calls for all nodes to output the number of nodes in the graph, called its size, and in the task of topology recognition each node has to learn the topology of the graph and its position in it.
In radio networks, nodes communicate in synchronous rounds and start in the same round. In each round a node can either transmit the same message to all its neighbors, or stay silent and listen. At the receiving end, a node $v$ hears a message from a neighbor $w$ in a given round, if $v$ listens in this round, and if $w$ is its only neighbor that transmits in this round. If more than one neighbor of a node $v$ transmits in a given round, there is a {\em collision} at $v$. We do not assume collision detection: in case of a collision, node $v$ does not hear anything
(except the background noise that it also hears when no neighbor transmits). The time of a deterministic algorithm for each of the above problems is the worst-case number of rounds it takes to solve it. Nodes have labels which are (not necessarily different)
binary strings. {Each node knows its own label and can use it when executing the algorithm.}
The length of a labeling scheme is the largest length of a label.

Our goal is to construct short labeling schemes for size discovery and topology recognition in arbitrary radio networks, and to design efficient deterministic algorithms for each of these tasks, using these short schemes.
It turns out that the optimal length of labeling schemes for both problems depends on the maximum degree $\Delta$ of the graph.
For size discovery, we construct a labeling scheme of length $O(\log\log\Delta)$ (which is known to be optimal, even if collision detection is available) and we design an algorithm for this problem using this scheme and working in time $O(\log^2 n)$, where $n$ is the size of the graph. We also show that time complexity $O(\log^2 n)$ is optimal for the problem of size discovery, whenever the labeling scheme is of optimal length $O(\log\log\Delta)$. For topology recognition, we construct a labeling scheme of length $O(\log\Delta)$, and we design an algorithm for this problem using this scheme and working in time 
$O\left(D\Delta+\min(\Delta^2,n)\right)$, where $D$ is the diameter of the graph. We also show that the length of our labeling scheme is asymptotically optimal, by proving that topology recognition in the class of arbitrary radio networks requires labeling schemes of length  $\Omega(\log \Delta)$.

\vspace*{0.5cm}

\noindent
{\bf keywords: } size discovery, topology recognition, radio network, labeling scheme

\end{abstract}

\pagebreak
	


	\renewcommand{\thefootnote}{\arabic{footnote}}

\section{Introduction}
\label{s:intro}
Information about the topology of the network or some of its parameters, such as size, often determines the efficiency and sometimes the feasibility of many network algorithms. For example, graph exploration with stop in rings with non-unique labels is impossible without knowing some upper bound on the size of the ring. On the other hand, optimal broadcasting algorithms in wireless networks with distinct labels are faster when the topology of the network is known than without this assumption \cite{CD2,KP2}. Hence, the problems of  \emph{size discovery} and \emph{topology recognition} are fundamental in network computing.
Size discovery calls for all nodes to output the number of nodes in the underlying graph, called its size, and in the task of topology recognition each node has to output an isomorphic copy of the graph with its position in it marked. More formally, in topology recognition, every node $v$ of the graph $G$ modeling the network must output a graph $G'$ and a node $v' $ in this graph, such that there exists
an isomorphism $f: G \rightarrow G'$, for which $f(v)=v'$. 

\subsection{The model}

We consider size discovery and topology recognition in radio networks modeled by simple undirected connected graphs. 
Throughout this paper $G=(V,E)$ denotes the graph modeling the network, $n$ denotes the number of its nodes, $D$ its diameter,  and $\Delta$ its maximum degree. 
We use square brackets to indicate sets of consecutive integers: $[i,j] = \{i, \dots, j\}$ and $[i] = [1, i]$.

As usually assumed in the algorithmic literature on radio networks, nodes communicate in synchronous rounds, starting in the same round. In each round a node can either transmit the same message to all its neighbors, or stay silent and listen. At the receiving end, a node $v$ hears a message from a neighbor $w$ in a given round, if $v$ listens in this round, and if $w$ is its only neighbor that transmits in this round. If more than one neighbor of a node $v$ transmits in a given round, there is a {\em collision} at $v$. Two scenarios concerning collisions were considered in the literature. The availability of {\em collision detection} means that node $v$ can distinguish collision from silence which occurs when no neighbor transmits. If collision detection is not available, node $v$ does not hear anything in case of a collision
(except the background noise that it also hears when no neighbor transmits). In this paper we do not assume collision detection. The time of a deterministic algorithm for each of the above problems is the worst-case number of rounds it takes to solve it. 

If nodes are anonymous then neither size discovery nor topology recognition can be performed, as no communication in the network is possible. Indeed, without any labels, in every round either all nodes transmit or all remain silent, and so no message can be received. Hence we consider labeled networks. A {\em labeling scheme} for a given network represented by a graph $G=(V,E)$ is any function $\cL$ from the set $V$ of nodes into the set $S$ of finite binary strings. The string $\cL(v)$ is called the label of the node $v$.
Note that labels assigned by a labeling scheme are not necessarily distinct. The {\em length} of a labeling scheme $\cL$ is the maximum length of any label assigned by it. Every node knows {a priori} only its label, and can use it as a parameter for the size discovery or topology recognition algorithm.

Our goal is to construct short labeling schemes for size discovery and topology recognition in arbitrary radio networks, and to design efficient deterministic algorithms for each of these tasks, using  such schemes. Such short schemes in the context of radio networks were studied for size discovery in \cite{GorainP18}, and for topology recognition in \cite{GorainP17}. In  \cite{GorainP18} the authors worked in the model with collision detection. They constructed labeling schemes of length $O(\log\log\Delta)$  and a size discovery algorithm using this scheme and working in time $O(Dn^2\log \Delta)$. They also proved that labels of size $\Omega(\log\log\Delta)$ are necessary to solve the size discovery problem in this model.
In \cite{GorainP17}, the authors studied topology recognition without collision detection, similarly as we do in the present paper, but restricted attention only to tree networks. They constructed labeling schemes of length  $O(\log\log\Delta)$ and a topology recognition algorithm working for arbitrary trees, using these schemes. Moreover, they showed that labels of size $\Omega(\log\log\Delta)$ are necessary to solve the topology recognition problem for trees.

Solving distributed network problems with short labels can be seen in the framework of
algorithms with {\em advice}.  In this paradigm that has recently got growing attention, an oracle knowing the network gives {advice} to nodes not knowing it,  in the form of binary strings, and a distributed algorithm cooperating with the oracle uses this advice to solve the problem efficiently.
The required size of advice (maximum length of the strings) can be considered a measure of the difficulty of the problem.
Two variations are studied in the literature: either the binary string given to nodes is the same for all of them \cite{GMP} or different strings may be given to different nodes
\cite{DBLP:conf/spaa/EllenGMP19,EllenG20,FKL,FPP}, as in the case
of the present paper. If strings may be different, they can be considered as labels assigned to nodes by a labeling scheme.
Such labeling schemes permitting to solve a given network task efficiently are also called {\em informative labeling schemes}.
One of the famous examples of using informative labeling schemes is to answer adjacency queries in graphs \cite{AKTZ}.

Several authors have studied the minimum amount of advice (i.e., label length) required to solve certain
network problems (see the subsection Related work). The framework of advice or labeling schemes permits us to quantify the amount of information used to solve a network problem, such as size discovery or topology recognition, regardless of the type of information that is provided.  It should be noticed that the scenario of the same advice given to all nodes would be trivial in the case of radio networks: no communication could occur, and hence the advice would have to contain the size of the network for size discovery, and would not help for topology recognition, as nodes would not be able to find their position in the network without communicating.

\subsection{Our results}
It turns out that the optimal length of labeling schemes,  both for size discovery and for topology recognition,  depends on the maximum degree $\Delta$ of the graph.
For size discovery, we construct a labeling scheme of length $O(\log\log\Delta)$, which is optimal, in view of \cite{GorainP18}, and we design an algorithm for this problem using this scheme and working in time $O(\log^2 n)$, where $n$ is the size of the graph. We also show that time complexity $O(\log^2 n)$ is optimal for the problem of size discovery, whenever the labeling scheme is of optimal length $O(\log\log\Delta)$. Hence, without collision detection we achieve the same optimal length of the labeling scheme, as was done in \cite{GorainP18} with collision detection, and for this optimal scheme our size discovery algorithm is exponentially faster than that in  \cite{GorainP18}.

For topology recognition, we construct a labeling scheme of length $O(\log\Delta)$, and we design an algorithm for this problem using this scheme and working in time 
$O\left(D\Delta+\min(\Delta^2,n)\right)$, where $D$ is the diameter of the graph. We also show that the length of our labeling scheme is asymptotically optimal, by proving that topology recognition in the class of arbitrary radio networks requires labeling schemes of length  $\Omega(\log \Delta)$. (In fact we prove a stronger result that this lower bound holds even in the model with collision detection.) If the optimal length of a labeling scheme sufficient to solve a problem is considered a measure of the difficulty of the problem, our result shows, in view of the labeling scheme of length $O(\log\log\Delta)$ for topology recognition in trees \cite{GorainP17}, that this task is exponentially more difficult in arbitrary radio networks than in radio networks modeled by trees. 

\subsection{Related work}

There is a vast literature concerning distributed algorithms for various tasks in radio networks. These tasks include, e.g., 
 broadcasting \cite{CGR,GPX}, gossiping \cite{CGR,GPPR} and leader election
\cite{CD,KP}. In some cases \cite{CGR,GPPR}, the authors use the model without collision detection, in others \cite{GHK,KP}, the collision detection capability is assumed.

Many authors use the framework of algorithms with advice (or equivalently informative labeling schemes) to investigate the amount of information needed to solve a given network problem.
In \cite{FIP1}, the authors compare the minimum size of advice required to
solve two information dissemination problems, using a linear number of messages.
In \cite{FKL}, it is shown that advice of constant size permits to carry out the distributed construction of a minimum
spanning tree in logarithmic time.
In \cite{AKTZ}, optimal labeling schemes are constructed in order to answer adjacency queries in graphs.
In \cite{EFKR}, the advice paradigm is used to solve online problems.

In the model of radio networks, apart from the previously mentioned papers  \cite{GorainP18} studying size discovery and \cite{GorainP17} studying topology recognition in the framework of short labeling schemes, other authors studied the tasks of broadcast  \cite{IKP,DBLP:conf/spaa/EllenGMP19,EllenG20}, and convergecast  \cite{BLPR20} in this framework. 
While  \cite{DBLP:conf/spaa/EllenGMP19,EllenG20} assume that short labels are given to anonymous nodes, \cite{IKP} adopts a different approach.
The authors study radio networks without collision detection for
which it is possible to perform centralized broadcasting in constant time, i.e., when the topology of the network is known and all nodes have different labels. They investigate
how many bits of additional information (i.e., not counting the labels of nodes) given to nodes are sufficient
for performing broadcast in constant time in such networks, if the topology of the network is not known to the nodes.

The task of topology recognition was also investigated in models other than the radio model: in \cite{FPP} the authors use the LOCAL model, and in \cite{MSRT} the model used is the congested clique.

\section{Preliminaries}\label{s:algorithmic:tools}


As our algorithms use recent results showing that there exist constant-length labeling schemes for the broadcast problem, we recall these results, adjust them to our needs and introduce an auxiliary notion called the {\em broadcast tree}. Then, using the fact that there exists an efficient broadcast algorithm with constant-length labels, we give a lemma allowing to ``encode'' a given message $M$ in labels of the neighborhood of some path associated with an efficient algorithm that collects the whole message in a single node and broadcasts it to the whole network.

\subsection{Broadcast algorithms and broadcast trees}\label{ss:algorithmic:broadcast}
We use the constant-length labeling schemes for the broadcast problem 
\cite{DBLP:conf/spaa/EllenGMP19,EllenG20}. In this problem, a source node $s$ has a message which must be communicated to all other nodes.
First, we recall the result regarding fast broadcast with constant-length labels from \cite{EllenG20}.
\begin{theorem}\label{t:broadcast}\cite{EllenG20}
	There exists a labeling scheme of length $O(1)$ and an algorithm $\executor$ using it which solves the 
	broadcast problem 
	in time $O(D\log n+\log^2n)$. 
\end{theorem}

We say that a node $v$ is \emph{informed} in some round $t$ of a broadcast algorithm if $v$ knows the broadcast message in round $t$. Otherwise, $v$ is \emph{uninformed} in round $t$.
We assume that feedback messages in the execution of $\executor$ are distinct from the broadcast message (this can be easily ensured by adding a special sign to the broadcast message).
\begin{definition}[Broadcast Tree] 
	Let  $G=(V,E)$ be a graph with the source node $s\in V$. 
	
	For each node $v\in V\setminus\{s\}$, the parent of $v$, denoted $\parent(v)$, is the first node which successfully transmits the broadcast message to $v$ during the execution of $\executor$ in $G$.
		The \emph{broadcast tree} $T_{\executor(G)}$ is the tree with the root $s$ and the set of edges $(v,\parent(v))$ for each $v\neq s$.
	
	The \emph{level} of a node $v\in V$, denoted $\text{level}(v)$, in the broadcast tree $T_{\executor(G)}$ is equal to the natural number $i$ such that $v$ receives the broadcast message from $\parent(v)$ for the first time in round $i$ of the execution of $\executor$ on $G$. 
\end{definition}

In our algorithms for the size discovery problem, we will use the following properties of the algorithm {\executor} from \cite{EllenG20}. 
\begin{lemma}\label{l:broadcast:properties}
	The broadcast algorithm $\executor$ described in Theorem~\ref{t:broadcast} 
	satisfies the following properties:
	\begin{enumerate}[(1)]
		\item
		Assume that the level of a node $v$ in the broadcast tree $T_{\executor(G)}$ is equal to $i$ and the level of $\parent(v)$ in this tree is $j<i$. Then, $\parent(v)$ has a child with level $k$ for each $k\in[j+1,i]$ such that $k\mod 3=1$.
		\item 
		The maximum value of the levels of nodes of the broadcast tree $T_{\executor(G)}$ is larger than $t-3$, where $t$ is the number of rounds of the execution of {\executor}.
		\item 
		Each node $v$ can determine the level of $\parent(v)$ and its own level in the broadcast tree $T_{\executor}(G)$.
	\end{enumerate} 
\end{lemma}
\begin{proof}
	We first give a short description of {\executor} \cite{EllenG20} focusing on elements which are important for the statement of the lemma. We refer the reader to \cite{EllenG20} for more details.
	The label of each node of the communication graph consists of three bits called \textsc{join}, \textsc{stay} nad \textsc{go}.
	The algorithm works in \emph{stages} such that the stage $i$ consist of three round $3i+1, 3i+2, 3i+3$ for a natural number $i$. Each node $v$ is active in a consecutive sequence of stages (possibly empty) starting at the stage directly following the stage at which $v$ receives the broadcast message. The labels guarantee that, if $v$ is active in a stage, at least one uninformed node receives the broadcast message from $v$ in the given stage.
	
	The sequence of sets $\text{DOM}_i$ for $i=1,\ldots, t$ plays an important role in the formal description and analysis of the algorithm, where $t$ is the number of stages of the algorithm.
	One can verify that the following properties are satisfied during an execution of the algorithm (see \cite{EllenG20}):
	\begin{enumerate}[(a)]
		\item 
		each node $v$ can determine whether it belongs to $\text{DOM}_i$ at the beginning of stage $i$, based on its label and its history of communication until stage $i-1$,
		\item 
		all elements of $\text{DOM}_i$ know the broadcast message already,
		\item 
		$\text{DOM}_i$ is nonempty for each $i\le t$,
		\item 
		each element of $\text{DOM}_i$ transmits the broadcast message in the first round of the stage $i$,
		\item 
		for each $v\in\text{DOM}_i$, at least one uninformed node $u\in \mathcal{N}(v)$ receives the non-ignored broadcast message from $v$ in the first round of stage $i$,
		\item 
		for each node $v$: either $v$ does not  belongs to $\text{DOM}_i$ for any $i$ or $v$ belongs to the sequence of sets
		$\text{DOM}_i, \text{DOM}_{i+1}, \ldots, \text{DOM}_j$ for $1\le i\le j<t$, where $t$ is the number of all stages.
	\end{enumerate}
	Below, we shortly describe the elements of the algorithm which are important with respect to the correctness of the above properties (a)--(f). 
	
	Let $\text{FRONTIER}_i$ be the set of nodes which are uniformed {before} stage $i$ but have at least one informed neighbor.
	The set $\text{DOM}_i$ (as DOMinating set) is always a \emph{minimal} set of informed nodes which dominates the set of frontier nodes {$\text{FRONTIER}_{i}$}.
	At first $\text{DOM}_1 = \{s\}$ and the conditions (a)--(f) are therefore satisfied at the beginning of stage~1.
	Assume that (a)--(f) are satisfied until the beginning of stage~$i$. Below, we describe an execution of stage $i$ and argue that the properties (a)--(f) are also preserved at the beginning of stage $i+1$.
	\begin{itemize}
		\item 
		In the first round of the stage $i$, each node in $\text{DOM}_i$ sends the broadcast message and the current round number $3i+1$, which assures (d). All newly informed nodes store the round number $3i+1$ of reception of the broadcast message. Each such newly informed node 
		becomes an element of the dominating set $\text{DOM}_{i+1}$ iff its \textsc{join} bit is equal to $1$.

		As $\text{DOM}_i$ is a \textbf{minimal} set dominating the set $\text{FRONTIER}_{i}$, i.e., the set of uninformed neighbors of informed nodes, each $v\in\text{DOM}_i$ has an (at least one) uninformed neighbor $u\in\text{FRONTIER}_{i}$ such that $v$ is the only neighbor of $u$ in $\text{DOM}_i$. 
		Therefore, such node $u$ receives the broadcast message from $v$ which gives (e). One of such newly informed neighbors of $v$ is chosen as the \textit{Feedback Node} of $v$ for stage $i$.
		
		\item
		The second round of the stage is the \emph{Feedback round}. In that round, each newly informed node checks the value of its \textsc{stay} and \textsc{go} bits. If at least one of them is equal to $1$, then the node transmits both of these bits. Each node from $\textsc{DOM}_{i}$ will no longer be in $\textsc{DOM}_{i+1}$ unless it receives $\textsc{stay}=1$ in the second round.

		\item 
		In the last round of the phase, introduced as a speed-up of the previous algorithm with constant size labels \cite{DBLP:conf/spaa/EllenGMP19}, a node from $\text{DOM}_i$ transmits the broadcast message if it got the value $\textsc{go}=1$ from its designated Feedback Node in the previous round.
		
		Finally, $\text{DOM}_{i+1}$ contains all elements of $\text{DOM}_i$ which received $\textsc{stay}=1$ in the second round of the stage as well as the nodes which were newly informed in rounds $1$ and $3$ of the stage $i$ and their $\textsc{join}$ bit is equal to $1$. This definition of $\textsc{DOM}_{i+1}$ proves the statements (a), (b) and (f).
	\end{itemize}
	Finally, if $\textsc{DOM}_i$ is empty, then all nodes are informed, as the communication graph is connected.
	And thus the broadcast task is finished, which gives the above property (c).
	
	
	Given the properties (a)--(f), one can see that the item (1) of the lemma directly follows from (e) and (f), the item (2) follows from (c) and the item (3) follows from the fact that each transmitting node attaches the current round number to its message transmitted in rounds 1 and 3 of a stage.
\end{proof}
Below, we show that {\executor} can be easily transformed into an acknowledged broadcast algorithm. We say that an algorithm $A$ solves the \emph{acknowledged broadcast} problem in $T$ rounds iff it solves the broadcast problem, and moreover, all nodes of the graph know after $T$ rounds that broadcast has been completed.

Similar ideas were applied for the acknowledged broadcast in \cite{DBLP:conf/spaa/EllenGMP19} (note that the formulation of acknowledged broadcast in \cite{DBLP:conf/spaa/EllenGMP19} was slightly weaker than ours). However, as the round complexity of the algorithm from \cite{DBLP:conf/spaa/EllenGMP19} is larger than the round complexity of {\executor} and we need additional properties of the algorithm for further applications, we provide the following corollary.
\begin{corollary}\label{c:ack:broadcast}
	Algorithm {\executor} can be transformed into algorithm {\execack} such that
	\begin{itemize}
		\item 
		{\execack} solves the {acknowledged} broadcast problem  on every graph $G$, using $O(1)$-bit labels and working in at most $3t$ rounds,
		where $t=O(D\log n+\log^2n)$ is the number of rounds of the execution of  {\executor} on $G$.
		\item after an execution of {\execack} on $G$, each node knows the number of levels of the broadcast tree $T_{\executor(G)}$, as well as its own level and the level of its parent in $T_{\executor(G)}$.
	\end{itemize}
\end{corollary}
\begin{proof}
	First, we use the labeling scheme of {\executor} for the \emph{standard} broadcast problem with $3$-bit labels.
	By Lemma~\ref{l:broadcast:properties}.(2), we can choose a path $P$ from the root to a leaf 
	in $T_{\executor(G)}$ with the last node $v_P$ of $P$ on the maximum level $l> t-3$ .
	Moreover, we add two extra bits to the labels of {\executor} in order to mark all nodes from $P$ and to mark the last node $v_P$ of $P$ (with the maximum level in $T_{\executor(G)}$). 
	
	Given the above described labels, the algorithm starts by an execution of {\executor} in order to broadcast the original message stored in the source node $s$. After reception of the broadcast message, $v_P$ knows that the execution of the broadcast algorithm is finished and it also knows the number of rounds $t$ 
	of {\executor}.
	Then, in round $t+1$, $v_P$ sends a new message containing the pair which consists of $t$ and the level of $\parent(v_P)$ (see Lemma~\ref{l:broadcast:properties}.(3)).\footnote{In order to make it clear that this message is not sent during the execution of {\executor} but in the next phase of the algorithm, one can add an appropriate prefix in messages, denoting the parts of the considered algorithm. We will omit this issue in further description of algorithms.} 
	Then, each other marked node $v$ transmits a message in the round following the reception of the pair $(t,\level(v))$. The message of $v$ contains $t$ and $\level(\parent(v))$.  The root $s$ receives the value of $t$ in at most $t$ rounds.
	After reception of $t$, the root starts the second execution of {\executor} with the value of $t$ as the broadcast message.
	Thus, all nodes receive $t$ as the broadcast message and they know that the execution of {\execack} is finished in $3t$ rounds.
\end{proof}

\subsection{Encoding a message in labels}\label{ss:algorithmic:encoding}
Now, we make an observation regarding the length of labels sufficient to store (in a distributed way) an arbitrary message $M$ initially unknown to all nodes, and subsequently make the message $M$ known to all nodes.
\begin{lemma}\label{l:message:on:path}
	Let $\execack$ be the acknowledged broadcast algorithm which
	 satisfies Corollary~\ref{c:ack:broadcast}.
	Then, an arbitrary message $M$ of size $m$ (initially unknown to nodes) can be made known to all nodes of graph $G$ in $O(t)$ rounds using $O(1+m/t)$-bit labels, where $t=O(D\log n+\log^2n)$ is the number of rounds of the execution of $\execack$ on $G$. 
\end{lemma}
\begin{proof}
	By Lemma~\ref{l:broadcast:properties}.(2) and Corollary~\ref{c:ack:broadcast}, we can choose a path $P$ from the root to a leaf 
	in $T_{\executor}$ with a leaf on the maximum level $L>t/3 -1$, where $t$ is the number of rounds of the execution of {\execack} on $G$.
	The tempting idea is to split the message $M$ into the labels of the path $P$ and collect these messages by transmissions from the end of $P$ backward in the direction of the root, similarly to the simple construction described in the proof of Corollary~\ref{c:ack:broadcast}.
	However, such a solution is insufficient, as the height of the broadcast tree can be much smaller that the number of levels of that tree (a node of $T_{\executor(G)}$ might have children of various levels).
	Therefore, we leverage the property from Lemma~\ref{l:broadcast:properties}.(1) and, for each node $v$ on the path $P$ and each level $k$ such that $v$ has at least one child of level $k$ in $T_{\executor(G)}$, we choose exactly one child of $v$ of level $k$.
	
	Then, we add one extra bit to the labels of {\execack} in order to mark all nodes from $P$ and their chosen children, as described above. 
	Apart from that, we use the labeling of $\execack$.

	Note that the number of marked nodes is at least $L/3>$ $t/9-1$, by Lemma~\ref{l:broadcast:properties}.(1).
	Thus, we can split the bits of $M$ between the marked nodes such that each node keeps at most $9\ceil{m/t}$ bits. 
	In order to make $M$ known to all nodes of the network, we use the following algorithm:
	\begin{enumerate}
		\item First,  $\execack$ is executed. At the end of this part each node $v$ knows the number of rounds of {\execack}, its own level $\level(v)$ and the level of its parent, as shown in Corollary~\ref{c:ack:broadcast}.
		
		\item Then, the message $M$ is gathered in the root $s$ of the broadcast tree.
		
		This part of the algorithm is initiated by the node $v_P$l. 
		Recall that $t$ is the number of rounds of the first part of the algorithm.
		Each marked node $v$ sends the tuple which consists of its part of the message $M$ and $\level(v)$ in the round $t+L-\level(v)+1$ together with all other messages received during this part of the algorithm.
		As at most one node from each level is marked, there is at most one transmitter in each round.
		Therefore all messages stored in the marked nodes will be collected in the root $s$ in round $t+L$, and the root can decode the whole message $M$ from these received messages.

		\item Finally, the root $s$ broadcasts $M$ to all nodes of the graph using {\execack} again.
	\end{enumerate}
\end{proof}


\section{The Size Discovery Problem}\label{s:size}

In Section~\ref{ss:sd:labeling} we describe a labeling scheme for encoding the size of a graph and an algorithm which makes use of this encoding to first collect information about the size of the graph at some node and then broadcast this value using an efficient broadcast algorithm. Then, in Section~\ref{ss:fast}, we describe a much faster size discovery algorithm with time complexity $O(\log^2n)$. Finally, in Section~\ref{ss:lower:sd}, we show that this time complexity is asymptotically optimal for labels of size $O(\log\log\Delta)$.
\subsection{Efficient labeling scheme}\label{ss:sd:labeling}
We design a general labeling scheme of length $O(\log\log\Delta)$ and a size discovery algorithm using this scheme, based on the broadcast algorithm working with constant-length labels \cite{EllenG20}.
In Section~\ref{sss:subtree} we prove a general lemma which states that each tree with maximum degree $\Delta$ and size $n$ contains a subtree with maximum degree $O(\log \Delta)$ and size $\Omega(\log n)$. Then, in Section~\ref{sss:labeling:scheme}, we describe a labeling scheme which yields a solution of the size discovery problem using the result from Section~\ref{sss:subtree}. Finally, in Section~\ref{sss:algorithm},
we design an algorithm which solves the size discovery problem using the labeling scheme from Section~\ref{sss:labeling:scheme} in time $O(D\log n+\log^2n)$, where $D$ is the diameter of the graph.


\subsubsection{$\log\Delta$-degree subtrees of a tree}\label{sss:subtree}
We prove a general lemma stating that, given a tree $T$ with maximum degree $\Delta$ and $3$ bits  of ``local memory'' at each node, and given a binary string $M$ of length $\lceil \log (n+1)\rceil+2$, there exists a subtree $T'$ of $T$ with maximum degree $\lfloor \log \Delta\rfloor+1$,
such that $M$ can be split into the $3$-bit local memories of the nodes of $T'$. 
\begin{lemma}\label{l:labeling}
	Let $T=(V,E)$ be a tree with maximum degree $\Delta$, size $n$ and the root $r\in V$. Let $M$ be a binary string of length $\lceil \log (n+1)\rceil+1$. Then, there exists a subtree
	$T'$ of $T$ rooted at $r$ and an assignment of binary strings to nodes of $T'$ such that:
	\begin{enumerate}[(a)]
		\item the number of children of each node of $T'$ is at most $\lfloor \log \Delta\rfloor+1$,
		\item the string assigned to each node of $T'$ has length at most $3$,
		\item the string assigned to the root of $T'$ has length $2$,
	\end{enumerate}
	and the concatenation of all strings assigned to nodes of $T'$ is $M$.
\end{lemma}
\begin{proof}
	In order to prove the lemma, we present a 
	recursive algorithm for assigning 
	bits of an arbitrary binary string $M$ of size  $\lceil \log (n+1)\rceil+1$ to the nodes of $T$ such that the assignment will satisfy the stated properties. Along with the construction, we present a proof of the lemma by induction on the size $n$ of the tree.
	
	For the basic step of induction, assume that $n=1$.
	Then $r$ has no children, $\ceil{\log(n+1)}=1$ and therefore the length of ${M}$ is at most $2$ -- we can store the bits of $M$ in the root $r$. 
	
	For the inductive step, assume that the lemma holds for all trees with size smaller than some $n>1$.
	Let ${M}$ be a binary string with at most $\ceil{\log{(n+1)}}+1$ bits and let $T$ be a tree with $n$ nodes, the root $r$ and the maximum degree $\Delta$. Let $v_0, \dots, v_{\floor{\log{\Delta}} }$ be the children of $r$ in $T$ 
	with the largest subtrees.
	Let $n_{v_i}$ denote the size of the subtree of $T$ rooted in $v_i$. We will construct an assignment $\mathcal{A}$ of substrings of $M$ to nodes such that the conditions (a)--(c) of the lemma are satisfied and the subtree $T'$ of nodes with assigned substrings contains the root $r$ and some nodes from the subtrees rooted at $v_0, \dots, v_{\floor{\log{\Delta}} }$.
	Split ${M}=$ $M[0]M[1]\cdots M[\ceil{\log(n+1)}]$ into blocks $B_0, \dots, B_{\floor{\log{\Delta}}}$ such that each block $B_i$ except the last one consists of $\ceil{\log(n_{v_i}+1)}+1$ bits of $M$ with indices in the range
	$$\left[\sum_{k=0}^{i-1} \left(\ceil{\log{(n_{v_k}+1)}}+1\right), \sum_{k=0}^{i} \left(\ceil{\log{(n_{v_k}+1)}}+1\right) - 1\right]$$
	provided that $\sum_{k=0}^{i} \left(\ceil{\log{(n_{v_k}+1)}}+1\right) - 1\le \ceil{\log(n+1)}+1$. Otherwise, the block $B_i$ contains all remaining bits of $M$ with indices larger than or equal to $\sum_{k=0}^{i-1} \left(\ceil{\log{(n_{v_k}+1)}}+1\right)$.\footnote{Note that, according to the above definition, some number of last blocks might be empty.}
	Then we recursively assign bits from the block $B_i$ to the tree rooted in $v_i$. As $n_{v_i}<n$, we can use the inductive hypothesis: as the length of $B_i$ is not larger than $\ceil{\log(n_{v_i}+1)}+1$, there exist assignments of bits of $B_i$  for $i=0,\ldots,\floor{\log\Delta}$ satisfying the conditions (a)--(c) of the lemma.
	In particular, the roots $v_0,\ldots,v_{\floor{\log \Delta}}$ have at most $2$ bits assigned. 
	In the case that the number $\sum_{k=0}^{\floor{\log\Delta}} \left(\ceil{\log{(n_{v_k}+1)}}+1\right)$ of bits
	assigned to subtrees rooted at $v_0,\ldots,v_{\floor{\log\Delta}}$ is smaller than $\ceil{\log(n+1)}+1$, we assign the remaining suffix of $M$ such that:
	\begin{itemize}
		\item the $i$th bit of the remaining suffix is assigned to $v_i$ for $i\in[0,\floor{\log\Delta}]$ (note that $v_i$ has been assigned only $2$ bits before, by (c));
		\item {up to two last remaining} bits are stored in the root $r$.
	\end{itemize}
	Now, it is sufficient to show that the number of assigned bits is at least $\ceil{\log{(n+1)}}+1$. 
	If the degree of the root $r$ is at least $\floor{\log\Delta} + 1$, then all tress rooted in marked children can store cumulatively $\sum_{k=0}^{\floor{\log\Delta}}\ceil{\log{(n_{v_k}+1)}}$ bits by the induction hypothesis. As we choose nodes with the biggest subtrees, $\sum_{k=0}^{\floor{\log\Delta}}n_{v_k} \geq \frac{\floor{\log\Delta}+1}{\Delta} n$. 
	Moreover, as described above, we can store additional $\ceil{\log\Delta} + 1$ bits in $v_0,\ldots,v_{\floor{\log \Delta}}$ and 2 bits in the root $r$. 
	In total this gives 
	\begin{align*}	
	\sum_{k=0}^{\floor{\log\Delta}} \ceil{\log{(n_{v_k}+1)}}  + \ceil{\log\Delta} + 1 + 2
	&\geq 
	\sum_{k=0}^{\floor{\log\Delta}}\log{(n_{v_k}+1)} + \log\Delta + 2
	\\ &\geq 
	\log\left(\prod_{k=0}^{\floor{\log\Delta}} (n_{v_k} + 1) \right) + \log\Delta + 2
	\\ &\geq  
	\log\left(\left(\sum_{k=0}^{\floor{\log\Delta}} n_{v_k}\right) + 1 \right) + \log\Delta + 2
	\\ &\geq 
	\log\left( \frac{\floor{\log\Delta+1}}{\Delta} n + 1 \right) + \log\Delta + 2
	\\ &\geq 
	\log\left( \frac{n+1}{\Delta} \right) + \log\Delta + 2
	\\ &\geq 
	\log{(n+1)} + 2  \geq \ceil{\log{(n+1)}}+1
	\end{align*}
	bits stored in the whole subtree.
	
	\noindent If the root $r$ of the tree has $x < \floor{\log\Delta} + 1$ children then we have in total
	\begin{align*}	
	\sum_{k=0}^{x-1} \ceil{\log{(n_{v_k}+1)}}  + x  + 2
	&\geq 
	\sum_{k=0}^{x-1}\log{(n_{v_k}+1)}  + 2
	\\ &\geq 
	\log\left(\prod_{k=0}^{x-1} (n_{v_k} + 1) \right) + 2 
	\\ &\geq  
	\log\left(\left(\sum_{k=0}^{x-1} n_{v_k}\right) + 1 \right) + 2 
	\\ &\geq 
	\log\left( n-1 + 1 \right) + 2
	\\ &\geq 
	\log\left( n \right) + 2
	\\ &\geq 
	\ceil{\log{(n+1)}} + 1
	\end{align*}
	bits stored in the whole subtree. These inequalities conclude the proof of the inductive step and thus give the proof of the lemma.


\end{proof}

\subsubsection{Labeling scheme}\label{sss:labeling:scheme}
In this section we describe the labeling scheme \textsc{CompactLabels} which combines the so-called $\Delta$-learning primitive from \cite{GorainP18}, the construction of a broadcast tree with help of the algorithm {\executor} from Theorem~\ref{t:broadcast}, 
and the limited-degree subtree described in Lemma~\ref{l:labeling}.

The labels will consist of the \textsc{root} bit and three disjoint blocks: \textit{$\Delta$-block}, \textit{BroadcastTree-block} and \textit{Message-block}. 

\noindent\textbf{$\Delta$-block and the \textsc{root} bit.}
\noindent First, we choose a node $r$ with the largest degree $\Delta$ in the graph $G$, and mark $r$ using the one-bit flag \textsc{root}. 
That is, $\textsc{root}_r=1$ and $\textsc{root}_v=0$ for each $v\neq r$.

Then, we describe $O(\log\log \Delta)$-bit strings called $\Delta$-blocks which will be used to learn the value of $\Delta$ by the node $r$.
The root $r$ is assigned the pair $(x, 0)$, where $x$ is the binary representation of the integer $\floor{\log\Delta} + 1$. This integer is the size of the binary encoding of $\Delta$. Then we choose $\floor{\log\Delta} + 1$ neighbors of $r$ and assign them consecutive natural numbers in the range $[1,\floor{\log\Delta} + 1]$. The $\Delta$-block of the $i$th chosen node is equal to the pair $(a_i,b_i)$ where $a_i$ is the binary representation of $i\in[\floor{\log\Delta}+1]$ of length $\floor{\log(\floor{\log\Delta} + 1)}+1$, with leading zeros,
 and $b_i$ is the $i$th bit of the binary representation of $\Delta$. 
The $\Delta$-blocks of other nodes are equal to the tuple $(0,0)$.
As the value of $\log\Delta$ can be encoded on $O(\log\log\Delta)$ bits, the length of $\Delta$-blocks is $O(\log\log\Delta)$.

\noindent\textbf{BroadcastTree-block.}
We apply Corollary~\ref{c:ack:broadcast} to construct the broadcast tree $T_{\executor}$ of the graph with root (source vertex) $r$ and assign $O(1)$-bit labels to nodes, used by the acknowledged broadcast algorithm {\execack}  working in $O(D\log n+\log^2n)$ rounds (Corollary~\ref{c:ack:broadcast}).

\noindent\textbf{Message-block.}
Let $M$ be the binary representation of  the size $n$ of the graph. The message-block is the concatenation of two blocks:  \textsc{index} and \textsc{message}.

We apply Lemma~\ref{l:labeling} in order to construct a subtree $T'$ of $T_{\executor}$ such that the degree of $T'$ is not larger than $\floor{\log\Delta}+1$ and each node of $T'$ is assigned a substring of $M$ of length at most $3$. 
{To each node $v$ of $T'$ apart from the root we assign a natural number $k_v$ such that the numbers assigned to the children of any node are consecutive integers starting from $0$.}
For a node $v$ of $T'$,  the block \textsc{index} of fixed length  $O(\log\log\Delta)$  is the binary representation of the integer $k_v\in[1,\floor{\log \Delta}]+1$ (with leading zeros), where $v$ is chosen as the child with number $k_v$ of its parent. The substring of length at most 3 assigned to each node of $T'$ according to Lemma~\ref{l:labeling} is the block \textsc{message} of this node.
The concatenation of these substrings of $M$ forms the string $M$ and substrings are assigned to nodes of $T'$ in post-order, i.e., the substring
assigned to a given node $v$ is situated in $M$ after the substrings assigned to the children of $v$ in $T'$, and substrings are assigned to the children $w$ in increasing order of $k_w$. 
The blocks \textsc{index} and \textsc{message} of nodes outside of tree $T'$ are the string $(0)$.

Conceptually, the label of a node is the concatenation of the root bit, the $\Delta$-block, the Broadcast Tree-block and the Message-block. In order to mark separations between the different blocks, we use the standard trick: the bit 1 is coded as 10, the bit 0 is coded as 01 and separations are coded as 00. This does not change the complexity of the label length, hence our labeling scheme has length
$O(\log\log \Delta)$.

\subsubsection{A simple size discovery algorithm}\label{sss:algorithm}

In this section we describe a simple size discovery algorithm using the above labeling scheme and working in time $O(D\log n+\log^2n)$. We will then show how to improve this algorithm to get the optimal time $O(\log^2n)$.

We first design the size discovery algorithm \textsc{AuxiliarySD},
which
%
uses the labeling scheme described in the previous section. The algorithm
	\textsc{AuxiliarySD} is a composition of three procedures, {\tt $\Delta$-learning}, {\tt Ack-broadcast}, and {\tt Size-learning}, corresponding to the three blocks of the labels described above: \textit{$\Delta$-block}, \textit{BroadcastTree-block} and \textit{Message-block}.
	
	\noindent
	Procedure {\tt $\Delta$-learning}.\\
	This procedure lasts $\floor{\log\Delta}+1$ rounds. In the $i$th round, the node with $\Delta$-block $(a_i,b_i)$ for $i>0$ transmits the message with two bits: $0$ and $b_i$. In each round all nodes except the root $r$ (i.e., the node with the \textsc{root} bit equal to $1$) ignore messages with the value of the first bit equal to $0$ and $r$ stores consecutive bits $b_i$. The root $r$ also knows the value of $\floor{\log\Delta}+1$ stored in its label, so it will know the number of rounds of the $\Delta$-Learning procedure.
	
	\noindent
	Procedure {\tt Ack-broadcast}\\ 
	We execute algorithm {\execack} from 
	Corollary~\ref{c:ack:broadcast}
	with the source vertex $r$, the labels from the BroadcastTree-block and the broadcast message equal to the binary string encoding the value of $\Delta$. 
	%
	As we use the algorithm solving the acknowledged broadcast problem, all nodes know the number of the last round of 
	this procedure.
	
	\noindent
	Procedure {\tt Size-learning}\\
	The goal of this procedure is to learn the binary representation $M$ of the size $n$ of the graph that is distributedly stored in nodes of the subtree $T'$ of the broadcast tree $T_{\executor}$, as described in the previous section. 
	
	Let $L$ be the number of levels of the broadcast tree $T_{\executor}$.
	The procedure {\tt Size-learning}  is split into phases $1,\ldots,L$ such that the strings \textsc{message} stored in nodes from the level $L-i+1$ are transmitted to their parents in phase $i$, together with messages containing strings \textsc{message} received by those nodes from their subtrees. Each phase lasts $\floor{\log\Delta} + 1$ rounds. 
	At the beginning of the phase $i$, each node $v$ of $T'$ of level $L-i+1$ reconstructs its message from all messages received in previous rounds and its own string \textsc{message}. 
	As the labeling scheme uses post-order encoding, the node $v$ first concatenates the messages of its children in $T'$ in the order of numbers assigned to them and then adds its own string \textsc{message} at the end of the string to get $M_v$, the part of $M$ stored in its subtree.
	Assume that the node $v$ of level $L-i+1$ has \textsc{index} that is the binary representation of the integer $k_v\in[1,\floor{\log\Delta}+1]$. Then, $v$ transmits $M_v$ in the round $k_v$ of phase $i$.
	
\begin{lemma}\label{l:sd:first}
	The algorithm \textsc{AuxiliarySD} solves the size discovery problem using a labeling scheme of length $O(\log\log \Delta)$, 
	and it works 
	in time $O(T_{\executor}(n,D)\cdot \log \Delta)$, where $T_{\executor}(n,D)=O(D\log n+\log^2n)$ is the time of the broadcast algorithm $\executor$ from Theorem~\ref{t:broadcast}.
\end{lemma}
\begin{proof}	
	Procedure {\tt $\Delta$-learning} works in time $O(\log \Delta)$. Time complexity of Procedure {\tt Ack-broadcast} is $O(T_{\executor}(n,D))$ by Corollary~\ref{c:ack:broadcast}. Procedure {\tt Size-learning} takes $O(\log\Delta)$ rounds for each round of  {\execack}. Thus, the claimed time complexity of \textsc{AuxiliarySD} follows from Theorem~\ref{t:broadcast} and Corollary~\ref{c:ack:broadcast}.
	The $O(\log\log\Delta)$ bound on the length of the labeling scheme follows from Corollary~\ref{c:ack:broadcast} and Lemma~\ref{l:labeling}.
	
	In order to prove correctness, it remains to show that all messages of nodes of $T'$ will be successfully delivered to their parents during Procedure {\tt Size-learning}.
	As each node $v$ has (at most) one child with the number $j\in[1,\floor{\log\Delta}+1]$ in its Message-block, children of $v$ will transmit in different rounds of the algorithm. Nevertheless, a child $v'$ of $v$ might collide with a child $u'$ of another node $u$ during an execution of Procedure {\tt Size-learning}, provided that $\level(u')=\level(v')$ and $k_{u'}=k_{v'}$.
	However, as $\level(u')=\level(v')$, the nodes $u'$ and $v'$ receive the broadcast message from their parents $u$, $v$ in the same round $l=\level(u')$. Moreover, their parents $u$, $v$ are different since $k_{u'}=k_{v'}$ and each node might have at most one child with a given value of \textsc{index}. Thus, reception of messages from $u$, $v$ by $u'$ and $v'$, respectively, in the round $l$ implies that there is no edge $(u,v')$ nor $(v,u')$ in the graph. This fact in turn contradicts our assumption that the messages of $u'$ and $v'$ cause a collision preventing delivery of their messages to the parents $u$ and $v$.
\end{proof}

Now, we combine Lemma~\ref{l:sd:first} with Lemma~\ref{l:message:on:path} to 
construct an improved version of \textsc{AuxiliarySD}, called \textsc{GeneralSD}.
The goal is to
get rid of the additional $\log\Delta$ multiplicative factor in the time complexity of a size discovery algorithm based on the broadcast algorithm {\executor}.

	Given a graph $G$, we first compute the number $t_G$ of rounds of the execution of {\executor} on $G$. If $t_G<\log n$, then we use the labeling scheme and the algorithm from Lemma~\ref{l:sd:first}.
As the number of levels of the tree $T_{\executor}$ is not larger than the number of rounds of the algorithm,
the number of rounds of the size discovery algorithm is $O(t_G \log\Delta)=$ $O(\log^2n)$.

If $t_G\ge \log n$, we use the constant-length labeling scheme and the broadcasting algorithm $\execack$ from Lemma~\ref{l:message:on:path}, with the binary representation of the size $n$ of the graph as the message $M$.
Thus, we obtain an algorithm which solves the size discovery problem in time $O(D\log n+\log^2n)$ using a labeling scheme of length $O\left(1+\frac{\log n}{t_G}\right)=$ $O(1)$.

Finally, in order to make the nodes of the graph aware of the chosen variant of the algorithm, we add one bit to all labels which contains the information whether or not $t_G<\log n$. Given the value of this bit, the nodes work according to instructions of the former or the latter algorithm described above. Thus, we obtain the following result.

\begin{theorem}\label{t:sz:first}
	The algorithm \textsc{GeneralSD} solves the size discovery problem using a labeling scheme of length $O(\log\log \Delta)$, and works 
	in time $O(D\log n+\log^2n)$.
\end{theorem}



\subsection{Fast algorithm for size discovery}\label{ss:fast}


Now, we design the algorithm \textsc{FastSD}, which solves the size discovery problem in time $O(\log^2n)$ using a labeling scheme with asymptotically optimal length $O(\log\log\Delta)$. Actually, if the diameter of the graph is $\Omega(\log n)$, constant-length labels are sufficient.

In Section~\ref{sss:multi}, we generalize the recent labeling schemes for broadcast to the multi-source broadcast problem. Then, in Section~\ref{sss:fast:idea}, we give the idea of our new size discovery algorithm, using the multi-source broadcast algorithm. Finally, in Section~\ref{sss:fast:details}, we give details of our algorithm with proofs of its correctness and efficiency.

\subsubsection{Multi-source broadcast}\label{sss:multi}
For the purpose of faster size discovery, we need a generalization of the result from \cite{EllenG20} regarding the multi-source broadcast problem.
Let $G=(V,E)$ be a graph  and let a non-empty set $S\subset V$
be the set of \emph{sources}. Assuming that all nodes from $S$ know the broadcast message $M$,
the goal of \emph{multi-source broadcast} is to deliver the message $M$ to all nodes of $G$.
Let the \emph{distance} from a node $v\in V$ to $S\subset V$ be the minimum of distances 
from $v$ to $s$ in $G$, over all $s\in S$. Moreover, let the diameter $D_S$ of $G$ with respect to $S$ be the maximum of distances from $v$ to $S$, over all $v\in V$.

We construct the multi-source broadcast algorithm \textsc{MBroadcast} through a simple modification of the 
broadcast algorithm $\executor$  from \cite{EllenG20}.
The labeling scheme for the (single-source) broadcast algorithm from \cite{EllenG20} is based on the probabilistic broadcast algorithm from \cite{Bar-YehudaGI92}. As one can see in Section~6.3 of \cite{EllenG20} (Lemmas~9 and 10), the probabilistic analysis depends merely on the length of a shortest path from a node which knows the broadcast message initially to a considered node $v$. 
{More precisely, the only differences with respect to {\executor} concern the definition of the sets $\text{FRONTIER}_1$ and $\text{DOM}_1$, i.e., the set of uninformed nodes that have an informed neighbor before round $1$ and the set of informed nodes that is a minimal dominating set of $\text{FRONTIER}_1$. As there is exactly one informed node $s$ at the beginning in the broadcast problem, we have $\text{DOM}_1=\{s\}$ and $\text{FRONTIER}_1$ is equal to the set of neighbors of $s$, in the labeling scheme used by {\executor}. For the multi-source broadcast problem, we 
	\begin{itemize}
		\item 
		set $\text{FRONTIER}_1=$ $\{v\,|\, (s,v)\in E\text{ and } v\not\in S\}$, 
		\item
		choose a subset of $S$ that is a minimum dominating set of $\text{FRONTIER}_1$ and let this chosen subset of $S$ be $\text{DOM}_1$.   
	\end{itemize}
}
Then, the analysis from \cite{EllenG20} works also for the multi-source broadcast problem and it implies a multi-source broadcast algorithm \textsc{MBroadcast} using a constant-length labeling scheme and working in time $O(D_S\log n+\log^2n)$.

\begin{corollary}\label{c:multi:broadcast} 
	Algorithm \textsc{MBroadcast} solves the multi-source broadcast problem using a labeling scheme of length $O(1)$, and works in time 
		$O(D_S\log n+\log^2n)$, where $S$ is the set of sources.
\end{corollary}

Given the above multi-source broadcast algorithm, we can describe 
our faster size discovery algorithm. For this purpose we introduce a few auxiliary notions.

Let $T_{\text{BFS}}$ be a BFS tree of the graph, where the root of the tree is an arbitrary node $s$. Moreover, let $V_i$ be the set of nodes at distance $i$ from the root, called the layer $i$ of $T_{\text{BFS}}$. Thus, in particular, $V_0=\{s\}$.

\subsubsection{The idea of the $O(\log^2n)$-time size discovery algorithm}\label{sss:fast:idea}
If the diameter $D$ of the graph $G$ is $O(\log n)$, we simply apply \textsc{GeneralSD} which gives the desired $O(\log^2n)$ time bound. 
%
Otherwise,
we conceptually split the graph $G$ into subgraphs of small diameter and then distribute information about the value of $n$ in these subgraphs separately. Still, we have to assure that collisions caused by edges connecting different subgraphs do not prevent successful execution of our subroutines in the considered subgraphs.  
%
%
{For brevity, we introduce the notation $\zet=\ceil{\log (n+1)}$}, i.e., $\zet$ is the length of the binary representation of the natural number $n$.
%
{We} say that a layer $V_i$ is \emph{green} iff $\lfloor i / \zet \rfloor$ is an even number; i.e., the layers $[j\zet,(j+1)\zet - 1]$ are green for each even number $j\ge 0$.
The nodes from layers $[j\zet,(j+1)\zet - 1]$ for even $j$ form the $j$th \emph{stripe} $X_j$.
Thus, we have the 0th, 2nd, 4th stripe and so on. Moreover, two consecutive stripes $X_j$ and $X_{j+2}$ are ``separated'' by nodes from the layers $V_{(j+1)\zet}, V_{(j+1)\zet+1}, \ldots,V_{(j+2)\zet -1}$ which do not belong to any stripe.
{Nodes from each stripe are called \emph{green}. Moreover, nodes from the last layer $V_{(j+1)\zet - 1}$ of the $j$th stripe for each even $j$ are called \emph{super-green}.}


Our algorithm {\fastsd} (an abbreviation of Fast Size Discovery) either executes \textsc{GeneralSD} (if $D=O(\log n)$) with appropriate labels or it consists of two stages:
\begin{itemize}
	\item[] \textbf{Stage~1}: In each stripe separately and in parallel, we execute an algorithm which ensures that, at the end of the stage, all super-green nodes from the stripe (i.e., nodes from the last layer $(j+1)\zet+1$ of the stripe $j$) learn the value of $n$ at the end of the stage. We perform this task in two phases. For the first phase, we choose a minimal set $U_j$ of nodes from the first layer of the stripe (called a minimal BFS-cover) such that each super-green node in the stripe $j$ is reachable from this set $U_j$, i.e., there is a BFS-path from a node of $U_j$ to this node. Then, for each element of $u\in U_j$ we choose a path of length $\zet$ such that it is a shortest path from $u$ to the last layer $(j+1)\zet -1$ of the stripe $j$. Moreover, we guarantee that those paths are conflict-free, i.e., there are no edges connecting nodes on different layers of different paths from the chosen family of paths. For each such path, the value of $n$ is then encoded in $1$-bit parts along the nodes of the path. In Phase~1 we ``collect'' these $1$-bit parts in appropriate nodes from the BFS-cover $U_j$ and, in Phase~2, we use the multi-source broadcast algorithm \textsc{MBroadcast} with $U_j$ as the the set of sources and $D=\log n$ in order to broadcast the size $n$ to super-green nodes (i.e., the nodes on the layer $(j+1)\zet -1$). As the consecutive stripes are ``separated'' by $\log n$ layers, we can execute Phase~1 separately in each stripe without the risk that interferences between nodes from different stripes cause any problem.

	\item[] \textbf{Stage~2}: Execute the multi-source broadcast algorithm \textsc{MBroadcast} (see Corollary~\ref{c:multi:broadcast}) with $S$ equal to the set of all super-green nodes, i.e., $S$ composed of the nodes from the layers $(j+1)\zet-1$, where $j$ is even and $(j+1)\zet- 1$ is not larger than the height of the BFS tree and $D_S=2\zet$. The broadcast message in the execution of \textsc{MBroadcast} is equal to the binary encoding of $n$, known to the nodes from $S$ after Stage~1. As each node of the graph is at distance $\le 2\zet$ from some super-green node, all nodes learn the value of $n$ in Stage~2.
\end{itemize}

\subsubsection{Details of algorithm \textsc{FastSD}}\label{sss:fast:details}
We say that a path $P=(x_1,\ldots,x_k)$ in a graph $G$ with a fixed source node $s$ is a \emph{BFS-path} if, for each consecutive nodes $x_i$ and $x_{i+1}$ on $P$, the layer of $x_{i+1}$ is {by one} larger than the layer of $x_i$. In other words, each edge of the path increases the distance from the source node $s$ by one. Thus, in particular, we do not use edges connecting nodes from the same layer. We say that $v$ is \emph{BFS-reachable} from $u$ if there exists a BFS-path from $u$ to $v$.

In order to implement Stage~1 of \fastsd, we first build a \emph{BFS-cover} 
of stripe $X_j$
by the nodes from the first layer $V_{j\log n}$ of the stripe, for each even $j$ (see Figure~\ref{f:stripe}). 

\begin{definition}[BFS-cover]
	For an even natural number $j$, the set of nodes $U_j\subseteq V_{j\log n}$ is a BFS-cover of the stripe $X_j=\bigcup_{i=j\zet}^{(j+1)\zet-1}V_i$ if, for each 
	node $v\in V_{(j+1)\zet-1}$, there exists a BFS-path $u_{0}, u_{1}, \ldots, u_{\zet-1}=v$ such that $u_i\in V_{j\zet+i}$ for each $i\in[0,\zet-1]$, and $u_0\in U_j$. In other words, the path $u_{0},  \ldots, u_{\zet-1}$ starts at some node from the given BFS-cover $U_j$ and each edge of the path goes to the layer with larger index.
\end{definition}
We say that a BFS-cover $U_j$ of the stripe $X_j$ is a \emph{minimal BFS-cover} of the stripe $X_j$ if no proper subset $U'$ of $U_j$ is a BFS-cover of $X_j$. A set of BFS-paths $P_1,\ldots,P_k$ 
is \emph{conflict-free} 
if, for each $i\neq j$, there is no edge $(x_i,x_j)$ in the graph such that $x_i\in P_i$, $x_j\in P_j$, $x_i$ and $x_j$ belong to different layers.

\begin{lemma}\label{l:cover:paths}
	Let $U_j=\{u_1,\ldots,u_k\}$ be a minimal BFS-cover of the stripe $X_j$. Then, there exists a set of \emph{conflict-free} BFS-paths $\{P_1,\ldots,P_k\}$ such that, for each $i\in[k]$:
	\begin{itemize}
		\item $P_i$ starts at $u_i$,
		\item the final node of $P_i$ is a super-green node of the stripe $X_j$,
	\end{itemize}
		
\end{lemma}
\begin{proof}
	As $U_j$ is a minimal BFS-cover, the following condition holds for each $l\in[k]$:
	there exists a (supergreen) node $y_l\in V_{(j+1)\zet-1}$ in the stripe $X_j$, such that $u_l$ and $y_l$ are connected by a BFS-path and there is no BFS-path from $u_i$ to $y_l$ for each $i\neq l$. Indeed, if this property does not hold for some $l\in[k]$, then $U_j\setminus\{u_l\}$ is also a BFS-cover which would contradict minimality of $U_j$. A node $y_l$ satisfying the above stated property will be called a \emph{witness} of $u_l$.
	
	Given the above property, we choose the paths $P_1,\ldots,P_k$ such that  $P_l$ is a BFS-path connecting $u_l$ with its witness $y_l$ and therefore there is no BFS-path from any other element $u_i\in U_j$ to $y_l$, for each $i\in[k]$, $i\neq l$.
	We claim that the paths $P_1,\ldots,P_k$ satisfy the properties stated in the lemma.
	For contradiction, assume that there exists an edge $(x_i,x_l)$ such that $x_i$ belongs to $P_i$, $x_l$ belongs to $P_l$, $i\neq l$, $x_i$ and $x_l$ are located in different layers. As the layers of nodes are determined by their depth in the BFS-tree $T_{\text{BFS}}$,
	each edge of the graph $G$ connects either nodes from consecutive layers or from the same layer.
	Thus, without loss of generality we can assume that $\layer(x_i)=\layer(x_l)+1$. This fact in turn implies that there exists a BFS-path from $u_l$ to $y_i$, the witness of $u_i$. Indeed, such a path contains edges of $P_l$ until $x_l$, then the edge $(x_l,x_i)$ followed by the edges of $P_i$ starting from $x_i$. This contradicts our choice of $y_1,\ldots,y_k$ as witnesses of $u_1,\ldots,u_k$. Thus the chosen set of paths $\{P_1,\ldots,P_k\}$ is conflict-free and satisfies the properties in the lemma.
\end{proof}

Let $M=M[1]M[2]\cdots M[\zet]$ be the binary representation of $n$.
In order to facilitate implementation of Stages~1 and 2 of \textsc{FastSD}, we construct labels of nodes in a given (green) stripe $X$ in the following way (see Figure~\ref{f:stripe}):
\begin{enumerate}
	\item A minimal BFS-cover $U=\{u_1,\ldots,u_k\}$ of the stripe $X$ is chosen, along with the set of conflict-free paths $\mathcal{P}=\{P_1,\ldots,P_k\}$ of length $\zet$, such that $P_i$ starts at $u_i$ and ends at {a node in} the last layer of $X$ (i.e., at a supergreen node).
	\item The set $X_{\text{BFS}}\subseteq X$ is determined such that $x\in X_{\text{BFS}}$ if $x$ is BFS-reachable from some element of the BFS-cover $U$. (Observe that, although each node from the last layer of a stripe is BFS-reachable from $U$ by the definition of a BFS-cover, it might be the case that nodes on smaller layers of the stripe are not reachable from $U$.)
	\item The label of each node $v$ is a concatenation of the following strings:
	\begin{enumerate}
		\item a 4-bit string composed of flags $\textsc{reach}_v$, $\textsc{super-green}_v$, $\textsc{cover}_v$, $\textsc{paths}_v$ indicating whether $v$ is green and it belongs to $X_{\text{BFS}}$ of its stripe $X$, whether $v$ is super-green, whether $v$ belongs to the chosen BFS-cover $U$ of its stripe, and whether $v$ belongs to one of the chosen cover-free paths $\mathcal{P}=\{P_1,\ldots,P_k\}$ 
		of its stripe, respectively;
		\item $M_v$: a one-bit string defined as follows. If $v$ is the $i$th node of a path $P_l$ from the above set of conflict-free paths, then $M_v$ is equal to $M[i]$; otherwise, the value of $M_v$ is 0;
		\item $B_v$: the constant-length label assigned to $v$ by the labeling scheme for the multi-source broadcast algorithm \textsc{MBroadcast}, with the graph $G'=(V',E')$, where $V'=X_{\text{BFS}}$,  $E'$ is the set of edges of $G$ connecting nodes from $X_{\text{BFS}}$, and the set of sources $S$ is equal to $U$.
		\item $\text{S2}_v$: the constant-length label assigned to $v$ by the labeling scheme for the multi-source broadcast algorithm \textsc{MBroadcast}, with the communication graph $G$, where the set of sources $S$ is equal to the set of all super-green nodes (i.e., such nodes $v$ that $\textsc{super-green}_v=1$).			
	\end{enumerate}
\end{enumerate} 


\begin{figure}[H]
	\centering
	\includegraphics[width=0.75\textwidth]{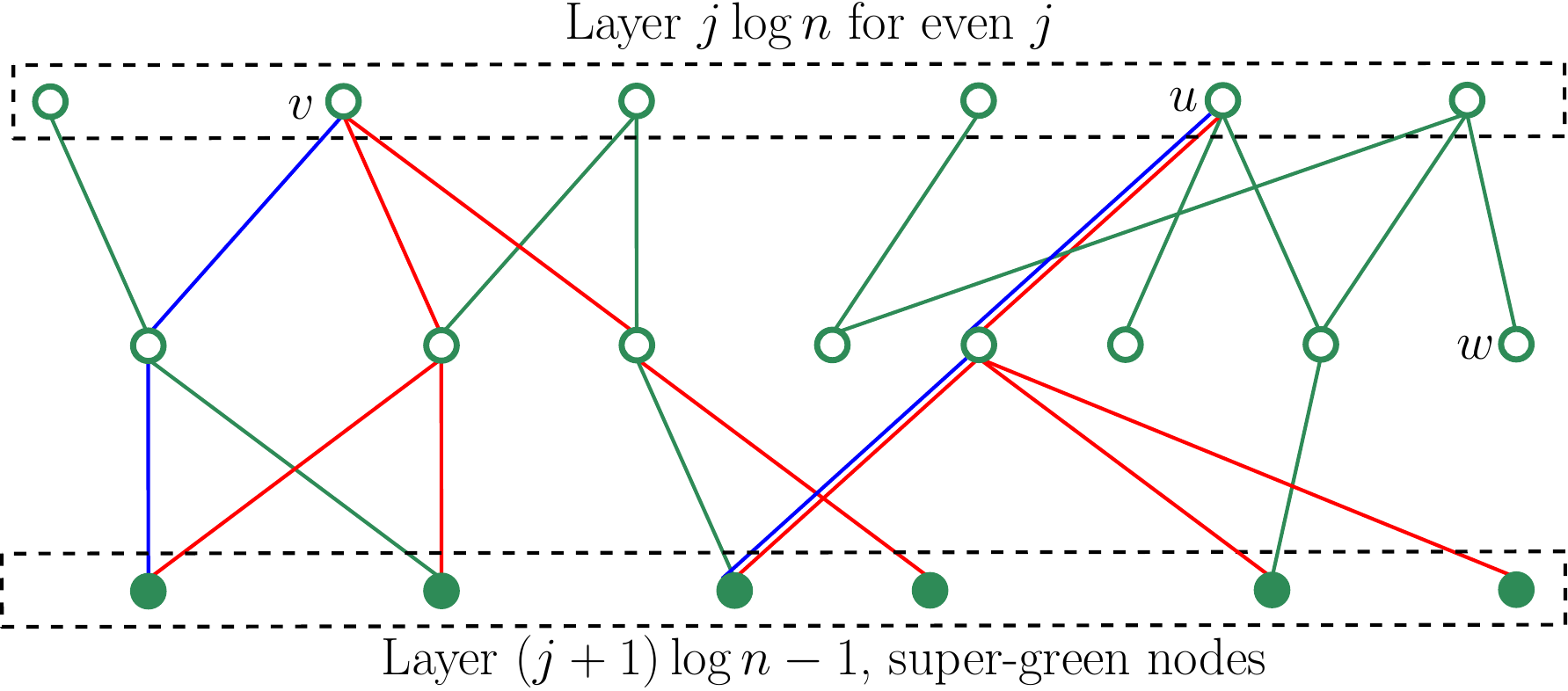}
	\caption{An example illustrating the notions and terms from the construction of \textsc{FastSD}. The edges denote paths of length $\frac12 \zet$. The set $\{u,v\}$ is a minimal BFS-cover -- the red edges show that each node from the layer $(j+1)\zet -1$ is reachable from $\{u,v\}$ by a BFS-path. The blue edges form a set of conflict-free BFS-paths described in Lemma~\ref{l:cover:paths}.  The node $w$ is an example of a node which is not active in Phase~2 of Stage~1, since $w\not\in X_{\text{BFS}}$. However, $w$ is at distance $\le 2\zet$ from some super-green node from the layer $(j-1)\zet - 1$ (stripe $j-2$) and therefore $w$ receives a message with the size of the graph in Stage~2. \label{f:stripe}}
\end{figure}

Using the above described labels, \textbf{Stage~1} for a stripe $X$ is implemented as follows:

\noindent\textbf{Phase~1:} collecting the value of $n$ at all nodes from the BFS-cover $U$. 
\begin{itemize}
	\item Round~1: each super-green  node $v$ from $\mathcal{P}$ (i.e., each $v$ such that $\textsc{super-green}_v=1$ and $\textsc{paths}_v=1$) sends $M_v$.
	\item Rounds $2, 3, \ldots, \zet$: each node $v$ with $\textsc{paths}_v=1$ which received a message $M$ in the previous round sends the concatenation of its bit $M_v$ and the received message $M$.
\end{itemize}
\noindent\textbf{Phase~2:} broadcast of the value of $n$ inside $X_{\text{BFS}}$. 

\noindent Using the labels $B_v$, all nodes with $\textsc{reach}_v=1$ execute the multi-source broadcast algorithm \textsc{MBroadcast}, where the set of sources $S$ is equal to $U$ and the graph is $G'=(V',E')$, $V'=X_{\text{BFS}}$ (determined by the flags $\textsc{cover}_v$) and $E'$ is the set of edges of $G$ connecting nodes from $X_{\text{BFS}}$.

\noindent\textbf{Stage~2} is an execution of the multi-source broadcast algorithm  \textsc{MBroadcast} on the whole graph $G$ using the labels $\text{S2}_v$ and thus with the source set $S$ consisting of all super-green nodes (i.e., the nodes $v$ with $\textsc{super-green}_v=1$).

\remove{\begin{lemma}\label{l:fast:sd}
	{\fastsd} solves the size discovery problem in $O(\log^2n)$ rounds using labels with size $O(1)$, provided that the diameter of the communication graph $G=(V,E)$ is larger than $\log n$. 
\end{lemma}
}
\begin{theorem}\label{l:fast:sd}
Algorithm	{\fastsd} solves the size discovery problem in time $O(\log^2n)$ using a labeling scheme of length
	$O(\log\log\Delta)$.
\end{theorem}
\begin{proof}
	As the case when $D=O(\log n)$ follows directly from Theorem~\ref{t:sz:first}, we assume that $D\ge \zet$.
	As there are no edges in a BFS-tree between nodes whose layers differ by more than one, computations in different stripes with green nodes do not interfere with each other, as the consecutive stripes are separated by $\zet$ layers. Moreover, as the set of paths $\mathcal{P}$ in each stripe is conflict-free, Phase~1 works without conflicts, and finally each node $u_i$ of the BFS-cover of the stripe (which is the source node of the path $P_i$) receives the whole message $M$ containing the binary representation of $n$, in $\zet$ rounds. 
	
	A potential problem with the implementation of Phase~2 might follow from the fact that (most of) the nodes do not know the actual value of $n$ or $\zet$ after Phase~1 and therefore they do not have access to information about the round in which Phase~2 starts. 
	Fortunately, 
	the multi-source broadcast algorithm requires only round synchronization of the set of source nodes $S$, cf.\ Corollary~\ref{c:multi:broadcast}.
	As the length of each path $P_i$ in the set of conflict-free paths $\mathcal{P}$ of each stripe is exactly $\zet$, the nodes from $U$ receive their messages during Phase~1 in the same round. As a minimal BFS-cover $U$ of a stripe is equal to the set of sources for Phase~2 of Stage~1, we can implement Phase~2 despite the fact that most nodes do not know the exact round number of the beginning of this phase.

	In order to synchronize the beginning of Stage~2, we can use the fact that all super-green 
	nodes know the value of $n$ at end of Phase~2 of Stage~1. Moreover, Phase~1 takes $\log n$ rounds and Phase~2 works in $O(D_S\log n+\log^2n)$ rounds, where $D_S=\zet$, as we execute  algorithm \textsc{MBroadcast} on the set of BFS-reachable nodes from $U$ in each stripe. Thus, the exact value of the upper bound $O(D\log n+\log^2n)$ on the time of an execution of Phase~2 is also fully determined by $n$. As the value of $n$ is known to the set of source nodes of an execution of the multi-source broadcast algorithm in Stage~2 (recall that the set of sources is equal to the set of super-green nodes of each stripe), those nodes are able to synchronize and start an execution of Stage~2 simultaneously.
	
	Recall that the set of super-green nodes is exactly the set of nodes of layers $j\zet$ for even values of $j$.
	Thus, each node of the graph is at distance at most $2\zet$ from a super-green node.
	Therefore, the multi-source broadcast algorithm executed in Stage~2 delivers the value of $n$ to all nodes of the graph.

	Finally, let us emphasize the importance of the assumption that the diameter of the graph is at least $\zet$. If the diameter of $G$ is smaller than $\zet$ then there are no super-green nodes and therefore we are not able to build conflict free BFS-paths with length $\zet$. This in turn excludes the possibility of implementation of Stage~1. Hence we treat the case when the number of layers is smaller than $\zet$ separately, by executing \textsc{GeneralSD}. {As the diameter $D$ of the graph is $O(\log n)$ in this case, we obtain time complexity $O(D\log n+\log^2n)=O(\log^2n)$.}
\end{proof}

\subsection{Lower bound on time complexity of size discovery}\label{ss:lower:sd}

In this section we show that the $O(\log^2n)$ time complexity of {\fastsd} is asymptotically optimal, as long as the length of the labeling scheme used by the algorithm is $o(\log n)$. To this aim, we leverage a powerful technical result from \cite{AlonBLP91} originally proved as a part of a proof of the $\Omega(\log^2n)$ lower bound on time complexity of broadcast in radio networks, even for centralized algorithms.

Let $G=(U\cup V,E)$ be a bipartite graph, where $U,V$ are the parts of the bipartition of $G$, i.e., $U\cap V=\emptyset$, 
and there are no edges between nodes inside $U$ nor inside $V$, such that 
$|U|=|V|=n$ and there is no isolated vertex in $U\cup V$.
A \emph{bipartite broadcast schedule} for $G$ is a sequence $U_1,\ldots, U_k$ of subsets of $U$ such that, if 
one executes a $k$-round radio network algorithm in $G$ with the set of transmitters in the $i$th round equal to $U_i$ for each $i\in[k]$, then
every node from $V$ receives (at least) one message during this execution {of the algorithm}. The number of subsets $k$ is called the \emph{size} of the bipartite broadcast schedule for $G$.

\begin{theorem}\label{t:ablp}\cite{AlonBLP91}
	There exists a constant $c>0$ such that, for each sufficiently large natural number $n$, there exists a bipartite graph $G$
	with sides of size $n$ such that the size of each bipartite broadcast schedule for $G$ is larger than $c\log^2n$.
\end{theorem}
Equipped with Theorem~\ref{t:ablp}, we are ready to prove the lower bound $\Omega(\log^2n)$ on the time complexity of the size discovery problem with short labels.

\begin{theorem}\label{time-lb}
	The time complexity of any algorithm solving the size discovery problem on all graphs with size at most $n$, using labeling schemes of length smaller than $\frac14 \log n$, is $\Omega(\log^2n)$.
\end{theorem}
\begin{proof}
	For the sake of contradiction assume that an algorithm $\mathcal{A}$ can solve the size discovery problem on all graphs of size at most $n$ in time at most $\frac{c}2 \log^2n$, using labeling schemes
	of length smaller than $\frac14\log n$,
	where $c$ is the constant from Theorem~\ref{t:ablp}. 
	For sufficiently large $m$, consider a sequence of bipartite graphs $G_1=(U_1\cup V_1,E_1),G_2=(U_2\cup V_2,E_2),\ldots,G_m=(U_m\cup V_m, E_m)$ such that the size of the sides of $G_i$ is $m+i$ and the size of each bipartite broadcast schedule for $G_i$ is larger than $c\log^2\left(2(m+i)\right)$. 
	A sequence $G_1,G_2,\ldots,G_m$ satisfying these conditions exists by Theorem~\ref{t:ablp}.
	Now, let $G'_i$ be the graph obtained by adding a new node $s_i$ to $G_i$, connected by an edge with all nodes from $U_i$.
	Observe that $G'_i$ is connected for each $i\in[m]$.
	Thus, 
	according to the above assumptions and by the fact that $\mathcal{A}$ works in time at most $\frac{c}2\log^2\left({2(m+i)+1}\right)$ on $G'_i$, there exists a node $v_i\in V_i$ for each $i\in[m]$, such that $v_i$ does not receive any message during an execution of $\mathcal{A}$ on $G'_i$. 
	 Nodes of each graph $G'_i$ have labels of size smaller than $\frac14 \log (2(m+i)+1)$. 
	Since $\frac14 \log(2(m+i)+1)\le \frac12 \log m$ for sufficiently large $m$, all labels in all graphs $G'_i$  have length at most $ \frac12 \log m$. Consequently,  there are at most $2\sqrt{m}$ distinct labels assigned to all nodes in all graphs $G'_i$. Hence there exist $i,j\in[m]$, $i\neq j$, such that the label of $v_i$ in $G'_i$ and the label of $v_j$ in $G'_j$ are equal, while $v_i$ and $v_j$ do not receive any message during executions of $\mathcal{A}$ on $G'_i$ and $G'_j$, respectively. It follows that nodes $v_i$ and $v_j$ must make the same decision concerning the size of graphs $G'_i$ and $G'_j$, respectively. However, the sizes of these graphs are different, which gives a contradiction.
\end{proof}

Since $O(\log\log \Delta)$ is $o(\log n)$ for any graph, and in view of the lower bound $\Omega(\log\log\Delta)$ from \cite{GorainP18} on the length of labeling schemes permitting size discovery, Theorems \ref{l:fast:sd} and \ref{time-lb} give the following corollary.

\begin{corollary}
Algorithm	{\fastsd} solves the size discovery problem on all graphs using a labeling scheme of asymptotically optimal length $O(\log\log\Delta)$, and it works in time $O(\log^2n)$, which is aymptotically optimal for size discovery algorithms using asymptotically optimal labeling schemes.
\end{corollary}

\section{Lower Bound on Lengths of Labels for Topology Recognition}\label{s:topology:lower}

We now start the investigation of the topology recognition problem.
In this section, we show that any algorithm for topology recognition
in general graphs requires a labeling scheme of length
$\Omega(\log \Delta)$. This length is thus exponentially larger than that sufficient to solve the easier problem of size discovery.
In fact, we will prove a stronger result: the above lower bound holds even in the model with collision detection. Hence, till the end of this section we assume that collision detection is available.

The proof is divided into two parts.
We first show the lower bound of $\Omega(\log n)$ for graphs of
 degree $2^{\Omega(\log n)}$.
Second, we generalize this result for graphs of arbitrary degree $\Delta$.

In our proof, we exploit the fact that each node in the graph
$G$ has to learn its degree.
Intuitively, if the set $\Set{deg(v)}{ v \in V}$ of node degrees in $G$
 is larger than the set $\Set{\mathcal{L}(v)}{ v \in V}$ of
node labels then
some two nodes
$u, v \in V,$ such that
$deg(u) \neq deg(v) \land \mathcal{L}(u) = \mathcal{L}(v)$
have to learn their degrees through their \textit{communication histories}
$\mathcal{H}$.
We will construct a family of graphs such that the labels
have to be of length {${\Omega(\log \Delta)}$}
so that nodes with the same labels and different degrees
could have different communication histories.

We now formally define a \textit{communication history} of a node for a given topology recognition algorithm executed on a given graph. Intuitively, this is a record of what the node learns during the execution of the algorithm (assuming collision detection).

\begin{definition}(Communication history) \label{D:Communication history}
	Let $G = (V, E)$ be a graph and let
	$\mathcal{L}$ denote a labeling scheme on $G$.
	Consider a topology recognition algorithm
	$\mathcal{A}$ executed on $G$ and let $i$ be a round of this algorithm.
	For each node $v \in V$,
	we define $\mathcal{H}_i(v)$ inductively
	such that $\mathcal{H}_i(v)$ denotes the communication history
	of $v$ at the end of the $i$th round of $\mathcal{A}$.
	
	\begin{itemize}
		\item $\mathcal{H}_0(v) = \emptyset$, for all $v$.
		
		\item If 
		\begin{itemize}
			\setlength\itemsep{0em}
			\item $v$ transmits in the $i$th round, or
			\item $v$ listens and more than one neighbor of $v$ transmits in the $i$th round (a collission), 
		\end{itemize}
		we append the special character \texttt{\#}, denoting "no message",
		to $v$'s history, i.e.,
		$\mathcal{H}_i(v) = [ \mathcal{H}_{i - 1}(v), \texttt{\#} ]$,

		\item If $v$ listens and none of its neighbors transmits in the $i$th round (silence), we append $\varepsilon$ to $v$'s history, i.e., $\mathcal{H}_i(v) = [ \mathcal{H}_{i - 1}(v), \varepsilon ]$.

		\item If $v$ listens and successfully receives a message
		$m \in \{0, 1\}^*$ in the $i$th round (exactly one neighbor of $v$ transmits $m$),
		we append $m$ to $v$'s history, i.e.,
		$\mathcal{H}_i(v) = [ \mathcal{H}_{i - 1}(v), m ]$.
	\end{itemize}
	
	\noindent
	$\mathcal{H}(v)$ denotes the communication history of $v$ at the end of the last round of the execution of 
	 $\mathcal{A}$ on~$G$.
\end{definition}


\subsection{The Lower Bound for Graphs with Large Degrees}\label{ss:lower:large:degre}

In this section we define a family  $\mathcal{G}$ of graphs of arbitrarily large size $n$ with
maximum degree $\Delta = 2^{\Omega(\log n)}$,
such that a labeling scheme of the length $\Omega(\log n)$ is
necessary to solve the topology recognition problem on each suffciently large graph
from this family.

Let $n$ be a sufficiently large natural number such that $n^{1/2}$ is an even
natural number.
We construct an $n$-node graph $G_n = (V, E)$.
The graph is composed of $n^{1/2}$ components, each component
is composed of 
$n^{1/2}$
nodes.
Let $C_i$ denote the set of nodes from the $i$th component.
Every node is connected with every node from a different component,
i.e., for every $C_i, C_j$ ($i \neq j$) and
every pair of nodes $u \in C_i, v \in C_j$ we have $\Tup{u}{v} \in E$.
Let $C(v)$ denote the component to which a node $v$ belongs.

We now describe the set of edges connecting nodes of
a component $C_i$ composed of $k=n^{1/2}$ nodes.
We divide nodes of $C_i$ into two sets $A$ and $B$, each of size
$k/2$.
Let $a_j \in A$ denote the $j$th node from $A$, and let $b_j \in B$ denote
the $j$th node from $B$.
Then, we connect $a_j$ with the first $j$ nodes from $B$:
$\Tup{a_j}{b_1}, \Tup{a_j}{b_2}, \dots, \Tup{a_j}{b_j} \in E$. This concludes the construction of $G_n$. Observe that each $G_n$ is connected.
Note also that for any component, the set of  different degrees of nodes in this component
has size 
$k/2$.
The degree of each node is in the range $[n - 2 n^{1/2}, n - n^{1/2}+1]$, since
each component has size 
$n^{1/2}$.
Note that 
nodes of different degrees that have equal labels and equal communication histories cannot
learn that they have different degrees.
The family $\mathcal{G}$ consists of all graphs $G_n$, such that $n^{1/2}$ is an even
natural number.

An example of a graph from the family $\mathcal{G}$ is presented in Figure~\ref{f:lower}.
\begin{figure}[H]
	\centering
	\includegraphics[width=0.75\textwidth]{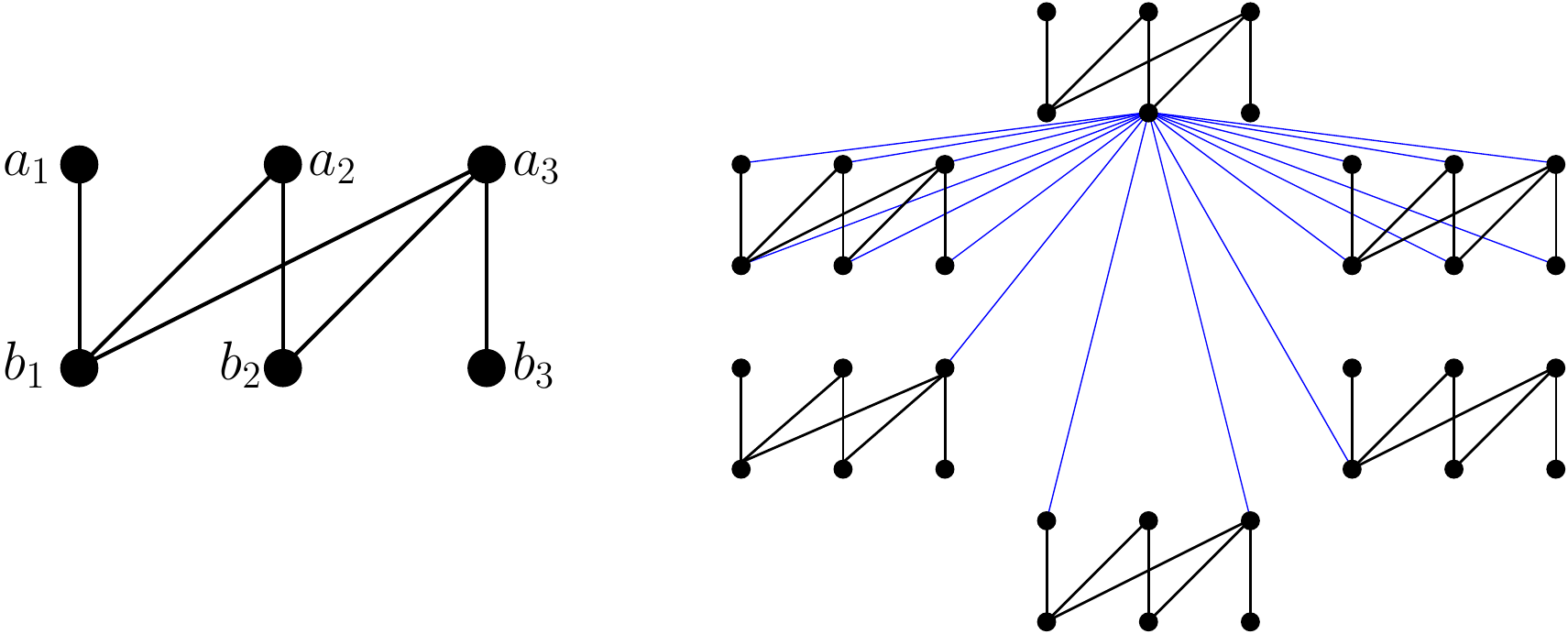}
	\caption{An illustration of the graph $G_{36}$ from the family $\mathcal{G}$. 
	The left picture depicts one component, the right picture depicts the whole graph. Blue edges connect nodes from different components. For clarity, only some of the inter-component edges are presented  (each pair of nodes from different components should be connected by an edge).  \label{f:lower}}
\end{figure}
Observe that, for each natural number $n$, there is at most one graph of size $n$ in the 
 family $\mathcal{G}$.

Below, we make a simple but useful observation regarding lack of information
received by most of  the nodes in graphs from $\mathcal{G}$ if at least two
nodes transmit in a round.

\begin{fact} \label{L:If 2 nodes transmit no other component hears}
	If at least two nodes $u\neq v$ transmit in the $i$th
	round then, {for every node $w$ 
	such that $C(w) \neq C(u) \land C(w) \neq C(v)$,} $w$ does not hear
	any message
	 in the $i$th round.
\end{fact}

We now define the \textit{canonical history}, for a given graph $G$ from the family $\mathcal{G}$ and a given topology recognition algorithm $\mathcal{A}$ executed on $G$.
The canonical history is a specific communication history, which will be common to many nodes of $G$ in the early stages of the execution of a topology recognition algorithm.
It has the property that a node with the
canonical history receives a message in a round $i$ iff every node in the whole graph receives
this message in this round. 

\begin{definition} (Canonical history) \label{D:Canonical history}

	Fix a graph $G$ from the family $\mathcal{G}$, with a given labeling scheme $\mathcal{L}$, and a topology recognition algorithm $\mathcal{A}$ executed on $G$.
	With respect to $G$, $\mathcal{L}$ and $\mathcal{A}$, 
	we define the canonical history $\widehat{\mathcal{H}}_i$ at the end of the $i$th round as follows:
	\begin{itemize}
		\item 
		$\widehat{\mathcal{H}}_0 = \emptyset$.
		\item $ \widehat{\mathcal{H}}_i = [ \widehat{\mathcal{H}}_{i-1}, m ] $,
		if there is exactly one node in $G$ that transmits
		the message $m$ in the $i$th round.
		\item $\widehat{\mathcal{H}}_i =
		[ \widehat{\mathcal{H}}_{i-1}, \texttt{\#} ] $,
		if 
		at least two nodes of $G$ transmit in round $i$,
		\item
		$\widehat{\mathcal{H}}_i =
		[ \widehat{\mathcal{H}}_{i-1}, \varepsilon ] $,
		if 
			no node of $G$ transmits in round $i$.
	\end{itemize}
\end{definition}
We then define the set of components sharing the canonical history.
\begin{definition} \label{D:Set of components sharing the canonical history}
Fix a graph $G$ from the family $\mathcal{G}$, with a given labeling scheme $\mathcal{L}$, and a topology recognition algorithm $\mathcal{A}$ executed on $G$.
	We say that a node $v$ \textit{has the canonical history}
	$\widehat{\mathcal{H}}_{i}$ if $\mathcal{H}_i(v) = \widehat{\mathcal{H}}_{i}$ 
	at the end of the $i$th 
	round of the execution of $\mathcal{A}$ on $G$.
	A component $C_j$ \textit{has the canonical history} 	at the end of the $i$th 
	round if every node
	from that component has the canonical history at the end of this round.
	
	Let $\widehat{\mathcal{C}}_i$ denote the set of components of $G$
	that have the canonical history $\widehat{\mathcal{H}}_{i}$ at the end
	of the $i$th round:
	$$ \widehat{\mathcal{C}}_i = \Set{C_j}{
		\forall_{v \in C_j} \mathcal{H}_i(v) = \widehat{\mathcal{H}}_{i}
	}. $$
\end{definition}

\paragraph{High-level idea of the proof of the lower bound.}
A node with the canonical history does not have any information
that distinguishes it from other nodes with the canonical history, apart from its label.  In topology recognition, nodes with different degrees must make different decisions (as each node has to situate itself in the graph), and there are at least $\frac12\sqrt{n}$ different degrees in each component of a graph $G_n \in \mathcal{G}$. 
The communication history of all nodes at the
beginning is equal to the canonical history. We show that, in every round, nodes from at most 
components can change their
histories from the canonical to a non-canonical  one (Fact~\ref{F:At most two components leave C}). 
Using this fact we show that some nodes change their histories from the canonical one to non-canonical 
{after}
at least $\frac12\sqrt{n}$ rounds.
We also show that a node makes such a change (from the canonical to non-canonical history) in round $i$ iff at least one node from its component (call it a \emph{trigger}) transmits a message in that round (Lemma~\ref{L:node needs to transmit when component leaves C}) and that 
{nodes from at most 
one other component transmit messages 
}
in that round (Fact~\ref{F:If three nodes transmit no one leaves C}). Using these properties we show that, for each label $l$, 
at most two nodes with the label $l$ can be triggers.
On the other hand, we need a trigger in each component, so $\frac14\sqrt{n}$ different labels are needed which implies that the length of the labeling scheme is $\Omega(\log n)$.

\paragraph{A detailed proof of the lower bound.}

Below, we present a formal proof of the lower bound $\Omega(\log n)$ on the length of a labeling scheme for topology recognition, following the above described ideas.

\begin{lemma} \label{L:Set of components with canonical is empty}
	Assume that a deterministic algorithm $\mathcal{A}$ solves the topology
	recognition problem in every graph $G_n$ from the family $\mathcal{G}$ using a labeling scheme $\mathcal{L}$ 
	with at most $n^{1/4}$ distinct labels (for $n > 256$).
	Let $f$ be the time of the execution of $\mathcal{A}$ on $G_n$.
	Then, the set of components that have the canonical history is empty
	after the last round, i.e., $\widehat{\mathcal{C}}_f = \emptyset$.
\end{lemma}

\begin{proof}
	Recall that every component in $G_n$ has size 
	$n^{1/2}$.
	For the sake of contradiction, assume that
	$\widehat{\mathcal{C}}_f$ is nonempty.
	Pick any $C_j \in \widehat{\mathcal{C}}_f$.
	\remove{Since each label has length at most $\frac{1}{4} \log n$, there
	are at most $2^{\frac{1}{4} \log n+1} = 2n^{1/4}$ different labels in $C_j$.
	Recall that the set of degrees of nodes in $C_j$ has size at least
	$\frac{n^{1/2}}{2}$.
	Hence 
	}
	Since there are at most $n^{1/4}$ distinct labels while the size of the set of degrees of nodes in $C_j$ is at least
	$\frac{n^{1/2}}{2}$,
	there are at least two nodes $u, v \in C_j$ with
	different degrees but with the same label and the same history
	$\widehat{\mathcal{H}}_{f}$. These nodes must make the same decision upon completion of $\mathcal{A}$ .
	Since  algorithm $\mathcal{A}$ solves the topology recognition problem, nodes $u$ and $v$ must output copies $R_u$ and $R_v$ of $G_n$ with their respective positions $u^*$ and $v^*$ marked,
	such that there exists an isomorphism $\phi$ from  $R_u$ to $R_v$ with $\phi(u^*)=v^*$. This contradicts the fact that $u$ and $v$ have
	different degrees in $G_n$.
\end{proof}

Let $\widehat{\mathcal{H}}$ be the sequence $(\widehat{\mathcal{H}}_{0}, \widehat{\mathcal{H}}_{1}, \ldots)$ 
of canonical histories for a given graph $G\in\mathcal{G}$ with a labeling scheme $\mathcal{L}$  and for a topology recognition algorithm $\mathcal{A}$. 
We say that a component $C_j$ leaves
$\widehat{\mathcal{H}}$ in the $i$th round
if $C_j \in \widehat{\mathcal{C}}_{i-1} \land
C_j \notin \widehat{\mathcal{C}}_{i}$.
If a component $C_j$ leaves $\widehat{\mathcal{H}}$
in the $i$th round then there is a node $v \in C_j$
such that
$\mathcal{H}_{i-1}(v) = \widehat{\mathcal{H}}_{i-1} \land
\mathcal{H}_{i}(v) \neq \widehat{\mathcal{H}}_{i}$,
i.e., $v$ stops having the canonical history in round $i$.
Let us discuss all possible cases when a component may leave
$\widehat{\mathcal{H}}$.

\begin{lemma} \label{L:node needs to transmit when component leaves C}
	A component $C_j$ may leave $\widehat{\mathcal{H}}$ in the $i$th
	round only in the case when at least one node from $C_j$ 
	and at most one node from some other component $C_k$ ($k \neq j$)
	transmit in the $i$th round.
\end{lemma}

\begin{proof}
	First, let us argue that there can be at most one node transmitting
	in some other component than $C_j$ when $C_j$ leaves $\widehat{\mathcal{H}}$ in the $i$th round.
	Indeed, if there are two different nodes transmitting in some other components
	$C_k, C_l$ ($k \neq j \land l \neq j$ but it is not necessary that $k\neq l$) then for
	every node $v$ in $C_j$ ($C_j \in \widehat{\mathcal{C}}_{i - 1}$) we have
	$\mathcal{H}_i(v) = [\widehat{\mathcal{H}}_{i-1}, \texttt{\#}] $.
	Moreover, if there are two nodes transmitting in the $i$th round then
	$\widehat{\mathcal{H}}_{i} = [\widehat{\mathcal{H}}_{i-1}, \texttt{\#}]$,
	according to Definition \ref{D:Canonical history}.
	
	Second, we will show that there must be at least one node from $C_j$
	transmitting in the $i$th round  when $C_j$ leaves $\widehat{\mathcal{H}}$ in the $i$th round.
	Pick any node $v \in C_j$ such that
	$\mathcal{H}_{i-1}(v) = \widehat{\mathcal{H}}_{i-1} \land
	\mathcal{H}_{i}(v) \neq \widehat{\mathcal{H}}_{i}$.
	Consider the case when some message $m$ was appended to the
	canonical history $\widehat{\mathcal{H}}_{i-1}$, i.e.,
	$\widehat{\mathcal{H}}_i = [\widehat{\mathcal{H}}_{i-1}, m]$.
	Then, by Definition \ref{D:Canonical history}, there is exactly one node
	transmitting in the $i$th round. If that node is not in $C_j$ then every
	node from $C_j$ receives the message $m$ and it does not leave $\widehat{\mathcal{H}}$ in the $i$th round.
	So, some node from $C_j$ must transmit in the $i$th round.
	Now, consider the case when
	$\widehat{\mathcal{H}}_i = [\widehat{\mathcal{H}}_{i-1}, \texttt{\#}]$.
	Then, by Definition \ref{D:Canonical history}, at least
	two nodes are transmitting in $G$ in round $i$.
	However, it follows from the fact proved in the previous paragraph that at most one of those nodes could not
	be in $C_j$. 
	Finally, consider the case when
		$\widehat{\mathcal{H}}_i = [\widehat{\mathcal{H}}_{i-1}, \varepsilon]$.
		Then, by Definition \ref{D:Canonical history}, 
		no node is transmitting in $G$ in round $i$.
		The set of nodes with the canonical history after round $i$ is equal to the set of nodes with the canonical history after round $i-1$ and therefore no component leaves $\widehat{\mathcal{H}}$ in the $i$th round.
		Hence, at least one node from $C_j$ must transmit in round $i$, provided that $C_j$ leaves $\widehat{\mathcal{H}}$ in round $i$.
	%
\end{proof}
Now, we state a few facts that will facilitate the proof of our lower bound. First, we observe
that nodes with equal labels and histories behave in the same way.
\begin{fact} \label{F:All nodes with the same label and hist transmit}
	Let $v$ be a node with label $\mathcal{L}(v) = l$
	such that $v$ has
	the canonical history at the end of round $i-1$, i.e.,
	$\mathcal{H}_{i-1}(v) = \widehat{\mathcal{H}}_{i - 1}$.
	If $v$ transmits in the $i$th round, then all nodes $u$
	such that $\mathcal{L}(u) = l \land
	\mathcal{H}_{i-1}(u) = \widehat{\mathcal{H}}_{i - 1}$ also transmit
	in the $ith$ round.
\end{fact}
Lemma~\ref{L:node needs to transmit when component leaves C} directly implies that no component can leave
$\widehat{\mathcal{H}}$  in a round in which nodes from more than two components send messages. 
\begin{fact} \label{F:If three nodes transmit no one leaves C}
	If there are at least three nodes in three different
	components transmitting in a round, then
	no component leaves $\widehat{\mathcal{H}}$ in that round.
\end{fact}
%
As a component can leave $\widehat{\mathcal{H}}$ only if at least one node from that component transmits (Lemma~\ref{L:node needs to transmit when component leaves C}) and no component leaves $\widehat{\mathcal{H}}$ in a round with transmitters from more than two components (Fact~\ref{F:If three nodes transmit no one leaves C}), we get the following fact.
\begin{fact} \label{F:At most two components leave C}
	In any given round at most two components may leave $\widehat{\mathcal{H}}$.
\end{fact}
Using the above facts and Lemmas~\ref{L:Set of components with canonical is empty}, \ref{L:node needs to transmit when component leaves C}, we are ready to prove a lower bound on the length of a labeling scheme necessary to solve the topology recognition problem.
\begin{theorem} \label{T:Log n lower bound}
	Any algorithm solving the topology recognition problem in a graph $G_n$ 
	from the family $\mathcal{G}$, for $n>8^4$, 
	requires
	a labeling scheme 
	with more than $n^{1/4}$ distinct labels.
\end{theorem}

\begin{proof}
	For a contradiction,  assume that there exists an algorithm $\mathcal{A}$
	that solves the topology recognition problem in a graph $G_n$ from the family $\mathcal{G}$, for $n>8^4$, using
	a labeling scheme $\mathcal{L}_n$ 
	with at most $n^{1/4}$ distinct labels.
	It follows from Lemma \ref{L:Set of components with canonical is empty}
	that all components leave $\widehat{\mathcal{H}}$ during the execution
	of $\mathcal{A}$ on $G_n$.
	
	Consider all rounds when some component leaves $\widehat{\mathcal{H}}$.
	Let $S_p\subseteq C_p$ be the subset of nodes of $C_p$ that transmit
	a message in the round when $C_p$ leaves $\widehat{\mathcal{H}}$.
	From Lemma \ref{L:node needs to transmit when component leaves C}
	we know that $S_p$ is nonempty.
	From Fact \ref{F:At most two components leave C}
	there are at most two components leaving $\widehat{\mathcal{H}}$ in any given
	round.
	Thus, the number of rounds when some component leaves $\widehat{\mathcal{H}}$ is
	at least 
	${\frac{n^{1/2}}{2}}$.
	Consider the $j$th round when at least one component
	leaves $\widehat{\mathcal{H}}$.
	Let $l_j$ denote the label of an arbitrary node $v$ such that $v$
	transmits in that round and $v$ is in the component that leaves
	$\widehat{\mathcal{H}}$ in that round.
	
	Let $t$ be the number of rounds in which some component(s) left
	$\widehat{\mathcal{H}}$. Consider the multiset
	$L = \{ l_1, l_2, \dots, l_t \}$.
	We claim that each label occurs at most twice in $L$.
	Indeed, assume that there are three indicies
	$i, j, k$ ($i < j < k$) such that
	$l_i = l_j = l_k$,
	i.e., there are at least three occurrences of the same
	label in $L$.
	Assume that $l_i, l_j, l_k$ are labels
	of nodes from components $C_i, C_j, C_k$, respectively.
	Assume that the round $i'$ was the $i$th round of the execution of $\mathcal{A}$ on $G$
	in which some
	component left $\widehat{\mathcal{H}}$, i.e.,
	$l_i$ is the label of a node that was transmitting in the
	$i'$th round.
	Recall that
	$C_i, C_j, C_k \in \widehat{\mathcal{C}}_{i' - 1}$, however,
	$C_i \notin  \widehat{\mathcal{C}}_{i'}$
	and $C_j, C_k \in \widehat{\mathcal{C}}_{i'}$.
	According to the definition of the sequence $L$, the choice of the values of $i$, $i'$ and in view of Lemma~\ref{L:node needs to transmit when component leaves C},
	some node $u \in C_i$ transmits in the round $i'$, such that
		$\mathcal{L}_n(u) = l_i$ and $u$ has the canonical history
	at least up to the $(i'-1)$st round, i.e.,
	$\mathcal{H}_{i'-1}(u) = \widehat{\mathcal{H}}_{i' - 1}$.
	%
	Our assumptions imply that there are two other nodes $v, w$ such that they have
	the same label as $u$ and the same canonical history up to the $(i'-1)$st round, i.e.,
	$$ v \in C_j \land w \in C_k, $$
	$$ \mathcal{L}_n(u) = \mathcal{L}_n(v) = \mathcal{L}_n(w), $$
	$$ \mathcal{H}_{i'-1}(u) = \mathcal{H}_{i'-1}(v) = \mathcal{H}_{i'-1}(w)
	= \widehat{\mathcal{H}}_{i' - 1}. $$
	Fact~\ref{F:All nodes with the same label and hist transmit}
	implies that all these nodes transmit in the $i'$th round.
	However, Fact \ref{F:If three nodes transmit no one leaves C}
	implies that no component leaves $\widehat{\mathcal{H}}$ in round $i'$, which
	contradicts our assumption that $C_i$ actually leaves $\widehat{\mathcal{H}}$ in this round.
	This contradiction 
	implies that each label occurs at most twice in $L$. Hence
	there are at least
	${\frac{n^{1/2}}{4}}$ 
	different labels in the multiset $L$.
	
	\remove{Since the algorithm $\mathcal{A}$ uses a labeling scheme on $G_n$ of length at most
	$\frac{1}{4} \log n$, it can assign at most
	$2^{\frac{1}{4} \log n+1}$ different labels.}
	Recall that the algorithm $\mathcal{A}$ uses a labeling scheme on $G_n$ 
		with at most $n^{1/4}$ distinct labels.
	However, for all $n > 8^4$ we have
	$\frac{n^{1/2}}{4} > n^{1/4}$. 
	This contradiction concludes the proof.
\end{proof}
\begin{corollary} \label{c:logn:lower:bound}
	Any algorithm solving the topology recognition problem in 
	all graphs from the family $\mathcal{G}$
	requires
	a labeling scheme of length 
	$\Omega(\log n)$
	on $G_n$.
\end{corollary}

\subsection{The Lower Bound for Graphs with Arbitrary Degrees}\label{ss:lower:topology:general}

In this section we show that, for all positive integers $\Delta<n$, there exists a graph
$H_{\Delta,n}$ of maximum degree $\Theta(\Delta)$ and of size $\Theta(n)$, such that any
algorithm solving the topology recognition problem in $H_{\Delta,n}$ 
requires a labeling scheme of length $\Omega(\log \Delta)$.
In order to prove this lower bound we use the family $\mathcal{G}$ of graphs $G_n$ of maximum degrees
$2^{\Omega(\log n)}$, where $n$ is the size of the graph, constructed in Section~\ref{ss:lower:large:degre}.

Choose arbitrary positive integers $\Delta<n$.
We construct a $\Theta(n)$-node graph $H_{\Delta,n}$ from
$\Theta(n / \Delta)$ isomorphic copies of a graph $G_k$ from $\mathcal{G}$, for some
$k=\Theta(\Delta)$.
Let $k$ be the smallest natural number such that $k\ge \Delta$ and $\sqrt{k}$ is a natural even number.
Let $G'_\Delta$ denote the graph obtained by the following slight modification of the graph $G_k$: add a new {\em special}
node $s$ connected to all other nodes of $G_k$.
Let $G'^{(i)}_\Delta$, for all $i = 1, 2, \dots, \ceil{\frac{n}{\Delta}}$,
denote the $i$th isomorphic copy of the graph $G'_\Delta$. 
All these copies are pairwise disjoint.

The graph 
$H_{\Delta,n}$
contains the above
$\ceil{\frac{n}{\Delta}}$ disjoint isomorphic copies $G'^{(i)}_\Delta$ of $G'_{\Delta}$, for all $i = 1, 2, \dots, \ceil{\frac{n}{\Delta}}$.
Moreover, 
we ensure connectivity of 
$H_{\Delta,n}$
in the following way.
Let $s_i$ denote the special node of $G'^{(i)}_\Delta$.
We connect every $s_i$ with $s_{i + 1}$
for each $i = 1, 2, \dots, \ceil{\frac{n}{\Delta}} - 1$ by an edge, as well as $s_{\ceil{\frac{n}{\Delta}}}$ with $s_1$.
Thus the special nodes of all graphs $G'^{(i)}_\Delta$ are connected in a ring.
As the graph $G'_\Delta$ is connected, the graph $H_{\Delta,n}$ is connected as well.
Moreover, $H_{\Delta,n}$ has maximum degree $\Theta(\Delta)$ and size $\Theta(n)$.

\begin{theorem}\label{lb-top}
	\remove{Any algorithm solving topology recognition problem in $\mathcal{G}$ requires
	a labeling scheme of the size greater than $\frac{1}{4} \log \Delta$.
}
	
		
	For all sufficiently large positive integers $\Delta<n$, 
	there exists a graph of size $\Theta(n)$ and of maximum degree $\Theta(\Delta)$ such that any topology recognition algorithm for this graph requires a labeling scheme 
	with more than $\Delta^{1/4}$ distinct labels.
\end{theorem}

\begin{proof}
	We will show that the above defined graphs $H_{\Delta,n}$ satisfy the statement of the theorem.
	


	
	In order to get a contradiction, assume that there exists an algorithm $\mathcal{A}$ that solves the topology
	recognition problem in 
	$H_{\Delta,n}$
	using a labeling
	scheme with at most  $\Delta^{1/4}$ distinct labels, for $\Delta>8^4$.
	Then, we can construct
	an algorithm $\mathcal{A}'$ that solves the topology recognition
	problem in 
	$G_{\Delta}$ with at most $\Delta^{1/4}$ distinct labels 
	as follows.
%
	
	Consider the component $G'^{(1)}_{\Delta}$ of $H_{\Delta,n}$.
	To each node $v$ of $G_{\Delta}$ we assign the label $(\mathcal{L}(v),E)$ such that:
	\begin{itemize}
		\item 
		The value of $\mathcal{L}(v)$ is equal to the label of the corresponding node of the component $G'^{(1)}_{\Delta}$ of $H_{\Delta,n}$.
		\item 
		The value of $E$ is equal to an encoding of all the messages sent by $s_1$ during an execution of $\mathcal{A}$ on $H_{\Delta,n}$ specifying the content of messages and rounds of transmissions of these messages.
	\end{itemize}
	Note that the value $E$ is common to all nodes $v$ of $G_{\Delta}$, and hence the number of distinct labels $(\mathcal{L}(v),E)$ is the same as the number of distinct labels $\mathcal{L}(v)$, i.e.,  at most $\Delta^{1/4}$.
	Using the above described labels, the algorithm $\mathcal{A}'$ works on $G_{\Delta}$ as follows. Each node $v$ works according to the algorithm $\mathcal{A}$, using the value $\mathcal{L}(v)$ as its label, with the following modification. After each round $t$, node $v$ modifies its communication history in this round using $E$, as follows. If $v$ sends a message in round $t$ or $v$ detects a collision in round $t$ or $s_1$ does not send a message in round $t$, then the communication history of $v$ in this round does not change. {However, if $v$ listens in round $t$, $v$ does not detect a collision in round $t$ and $s_1$ sends a message $m$ in round $t$, then $v$ replaces the term of its communication history in round $t$ by: 
	\begin{itemize}
		\item the message $m$, if $v$ does not receive any message in round $t$,
		\item the collision symbol \texttt{\#}, if $v$ does receive a message in round $t$.
	\end{itemize}
	} (Intuitively, $v$ simulates the message that it would receive from $s_1$ in the graph 
	$H_{\Delta,n}$ or the collission which would result from the transmission of $s_1$.) In the next round $v$ uses the modified communication history and executes algorithm $\mathcal{A}$.
	Hence, for every round $t$, the communication history of node $v$ when the above defined algorithm is executed on $G_\Delta$ is the same as the communication history of the corresponding node in
	 $G'^{(1)}_{\Delta}$ when $\mathcal{A}$ is executed on $H_{\Delta,n}$. Using the assumption that $\mathcal{A}$ solves the topology recognition problem in $H_{\Delta,n}$, we conclude that $\mathcal{A}'$ solves the topology recognition problem in $G_\Delta$ for $\Delta>8^4$. This is a contradiction with Theorem~\ref{T:Log n lower bound} that concludes the proof.

\remove{	All of the messages sent by the special node $s_i$ to nodes of $G'^{(i)}_\Delta$,
	during the execution of $\mathcal{A}$ on $H_{\Delta,n}$ 
	can be encoded as a part of the algorithm $\mathcal{A}'$,
	so
	that communication histories of all nodes of $G'^{(i)}_{(\Delta)}$, except $s_i$, during the execution of
	$\mathcal{A}$ on 
	$H_{\Delta,n}$
	are the same as during the execution
	of $\mathcal{A}'$ on $G'^{(i)}_{(\Delta)}$. Recall that nodes from $G'^{(i)}_\Delta$, excluding $s_i$, do not have neighbors outside of this subgraph.
}
\end{proof}

The following corollary of Theorem \ref{lb-top} is the main result of this section.

\begin{corollary}\label{cor-main}
For all positive integers $\Delta<n$, there exists a graph
 of maximum degree $\Theta(\Delta)$ and of size $\Theta(n)$, such that any
algorithm solving the topology recognition problem on this graph
requires a labeling scheme of length $\Omega(\log \Delta)$.
\end{corollary}

\section{Topology Recognition: Labeling and Algorithm}\label{s:topology:algorithm}
Our final result is a topology recognition algorithm working for any graph of maximum degree $\Delta$ using a labeling scheme of length $O(\log \Delta)$. This length is optimal, in view of Corollary \ref{cor-main}.
Our algorithm works in time  $O\left(D\Delta+\min(n,\Delta^2)\right)$.
First, in Section~\ref{ss:BFS:broadcast:gather} we focus on constructing a BFS tree of the 
 graph and efficient broadcast/gathering algorithms working in this tree.
Then, in Section~\ref{ss:tp:algorithm}, the main topology recognition algorithm is presented, 
using the broadcast and gathering
subroutines from Section~\ref{ss:BFS:broadcast:gather}.

\subsection{Broadcast-gathering primitive}\label{ss:BFS:broadcast:gather}

In this section we describe a labeling scheme of length $O(\log \Delta)$ and algorithms for the (acknowledged) {\em broadcast} and {\em gathering} problems using this scheme. Both algorithm work in time $O(D \Delta)$.


Consider a graph $G = (V, E)$ of maximum degree $\Delta$ and diameter $D$, with a root node $r \in V$. 
\remove{
For the \texttt{broadcast} problem we assume that the root node $r$ has some massage $M$. At the end of an execution of a broadcast algorithm all nodes must know $M$. }
\remove{
A solution of the \texttt{Acknowledged Broadcast} problem should guarantee that, apart from the standard broadcast, the root node $f$ is aware of the round number at which the algorithm is finished.
}
In the {gathering} problem, each node $v \in V$ has some message $M_v$ and the root $r$ has to learn all messages $M_v$. 



%
Denote by $\textit{layer}(v)$ the distance from $r$ to a node $v$. 
We say that $v$ is at layer $i$ if $\layer(v)=i$.
Let $V_i$ be the set of nodes at  layer $i$. 
The neighborhood of a node $v$ in $G$ is denoted by $\mathcal{N}(v)$ and the set of neighbors of $v$ at layer $i$ is denoted by $\mathcal{N}_i(v)$. We fix an arbitrary strict total ordering on nodes denoted by $\prec$.
For each node $v$, its label $\mathcal{L}(v)$ is equal to the tuple 
$(r_v, \leaf(v), a_v, b_v, g_v, \Delta)$, where $r_v$ is the $1$-bit flag indicating whether $v$ is the root node $r$, $\leaf(v)$ is the bit indicating whether $v$ has any neighbors on higher layers than $\layer(v)$, $a_v$ is the bit used for acknowledged broadcast, while both $b_v$ and $g_v$, called \emph{broadcast-label} and \emph{gather-label}, are special integer values used for broadcast and gathering respectively, described in more detail below. More precisely, the label $\mathcal{L}(v)$ is the concatenation of bits $r_v$, $\leaf(v)$, $a_v$, and binary representations of integers $b_v$, $g_v$ and $\Delta$.


\paragraph{The idea of the broadcast/gather labels and algorithms.}
For the broadcast problem, we split the nodes of each layer $i$ into subsets $X_0,\ldots,X_{\Delta-1}$, 
{where}
$X_j$ is a maximal set of nodes of layer $i$ not belonging $\bigcup_{k<j} X_k$ such that the sets of neighbors of $X_j$ at layer $i+1$ are pairwise disjoint, exluding the neighbors of nodes from $\bigcup_{k<j} X_k$. For $v\in X_j$ we set $b_v=j$. We observe that all nodes from $X_j$ can transmit in the same round in order to deliver the broadcast message to all their neighbors at layer $i+1$, excluding the neighbors of the nodes from $\bigcup_{k<j} X_k$. This fact makes possible to deliver the broadcast message from layer $i$ to layer $i+1$ in $\Delta$ rounds. For each node $v$ at layer $i+1$, the first node which successfully transmits the broadcast message to $v$ becomes $\parent(v)$ and $v$ becomes its child.

For the gathering problem, we assign distinct gather-labels in the range $[0,\Delta-1]$ to the children of each node.
This assignment is required to satisfy also the following additional restriction. If $u$ at layer $i+1$ is a neighbor of $v$ at layer $i$ and $u$ is not a child of $v$, then the gather-label of $u$ is not assigned to any child of $v$. With this restriction, the nodes at level $i$ with a given value of the gather-label can transmit successfully messages to their parents in the same round. Thanks to this property, the values/messages stored in the nodes from the layer $i>0$ can be gathered in their parents at layer $i-1$ in $O(\Delta)$ rounds.

More details, formal statement of the above mentioned properties with their proofs and the final broadcast and gathering algorithms are described below.

\subsubsection{Assignment of broadcast-labels and gather-labels}
\paragraph{Assignment of the values $b_v$ and construction of a BFS-tree} 
For each layer $i$, we assign the values $b_v$ as follows. First, we take any maximal subset $X_0$ of $V_i$ such that the sets of neighbors at layer $i+1$ of elements of this set are pairwise disjoint. 
All nodes from $\mathcal{N}_{i+1}(v)$ will be called \textit{children} of $v$, for each $v\in X_0$. Then, we construct sets $X_1, X_2, \dots, X_{\Delta-1}$ as follows. 
Assume that the sets $X_0, \dots, X_{j-1}$ are already constructed for $j>0$. 
We define $X_j$ as a maximal subset of nodes $v$ at layer $i$ such that 
\begin{enumerate}[(a)]
	\item $v \notin X_0\cup\dots\cup X_{j-1}$,
	\item for each pair of nodes $v\neq u$ such that $v,u \in X_j$, we have $$ \mathcal{N}_{i+1}(v) \cap \mathcal{N}_{i+1}(u) \setminus \left(\bigcup_{k=0}^{j-1} \bigcup_{x\in X_k} \mathcal{N}_{i+1}(x)  \right) = \emptyset.$$ 
\end{enumerate} 
If a node $v\in X_j$ has a neighbor $u$ in $V_{i+1}$, such that no node from $X_k$ is in $\mathcal{N}_i(u)$ for all $k<j$, we say that $v$ is the \emph{parent} of $u$ and $u$ is a \emph{child} of $v$. Observe that the set of edges determined by this parent-child relationship will form a BFS tree of the graph, denoted $T_{\text{BFS}}$.

Below, we show that all nodes at layer $i$ belong to $X_0\cup\cdots\cup X_{\Delta}$.
\begin{lemma}
For each layer $i$, the union of the sets $X_0\cup\cdots\cup X_{\Delta}$ from this layer is equal to the set of all nodes at layer $i$.
\end{lemma}

\begin{proof}
	
	Let us check the sizes of the sets $Z_j(v) =  \mathcal{N}_{i+1}(v) \setminus \left(\bigcup_{k=0}^{j-1} \bigcup_{x\in X_k} \mathcal{N}_{i+1}(x)  \right)$ for each value of $j\in[0,\Delta]$. For $j=0$ the size of $Z_j(v)$ is at most $\Delta$. For other values of $j$ we have 
	
	\begin{itemize}
		\item Case~1: $v \in X_k$ for some $k<j$: then $Z_j(v)$ is empty.
		\item Case~2: $v \notin X_k$ for all $k<j$. 
		
		From the maximality of $X_{j-1}$ we know that there exists $u\in \mathcal{N}_{i+1}(v)$ which is also in $\mathcal{N}_{i+1}(v')$ for some $v' \in X_{j-1}$ and not in neighborhoods of any node from $X_0, \dots, X_{j-2}$. 
		This fact implies that $|Z_j(v)| \leq |Z_{j-1}(v)| + 1$ for each $0\le j\le k$.
	\end{itemize}
The above properties imply that, for each node $v\not\in X_0\cup\cdots\cup X_{j-1}$ at layer $i$, the size of
$Z_j(v)$ is at most $\Delta-j$. This in turn implies that $Z_\Delta(v)$ must be empty for each $v\not\in X_0\cup\cdots\cup X_{\Delta-1}$. 
So each node $v$ at layer $i$ is in some set $X_0, X_1, \dots, X_{\Delta}$.
\end{proof}
We set $b_v\gets k$ for each node $v\in X_k$, where the sets $X_0,\ldots,X_{\Delta}$ are as above.


For the acknowledged broadcast, we choose some arbitrary leaf node with the highest value of $\layer(v)$ and a path $P$
of length $\layer(v)$, such that $P$ starts at $r$ and ends at $v$ in the tree $T_{\text{BFS}}$ obtained through parent-child relationships described above. The value of $a_u$ is set to $1$ for each node $u$ on the path $P$ 
while the value $a_w$ is set to $0$ for each node $w$ outside of $P$.


\paragraph{Assignment of the values $g_v$.} 
As mentioned before, the bit $g_v$ of each node $v$ will be used in our gathering algorithm.
The value $g_r$ of the root node $r$ is equal to $0$. 
Then we assign values $g_v$ layer by layer, assigning the values to nodes at layer $i$ in phase $i$.
In the $i$th phase, we order nodes at layer $i-1$ according to the ordering determined by the values of $b_v$, from the smallest to the largest and, among nodes with the same value of $b_v$, according to the ordering $\prec$. Let $v_j$ be the  $j$th node at layer $i-1$ in this order. We give the value $g_u$ to each \textit{child} $u$ of $v_j$ in the BFS-tree  $T_{\text{BFS}}$, as follows.
For a node $v$ at layer $i-1$, we order its children at layer $i$ (i.e., its neighbors at layer $i$ which have not been assigned gather-labels yet) according to $\prec$, and assign to each of them the smallest  non-negative integer not assigned earlier to any neighbor of $v$ at layer $i$.

Below, we state some properties of gather-labels which will be relevant fo the gathering algorithm presented in Section~\ref{ss:bandg:bfs}.
\begin{lemma}\label{l:gv:bits}
	The integers $g_v$ have the following properties:
	\begin{enumerate}[(a)]
		\item $g_v\in[0,\Delta-1]$,
		\item $g_v\neq g_u$ for each $u\neq v$ such that $\parent(u)=\parent(v)$,
		\item if $\parent(v)\neq\parent(u)$ and $g_u=g_v$ for some $u\neq v$ at the same layer, then $u$ and $\parent(v)$ are not connected by an edge.
	\end{enumerate}
\end{lemma}
\begin{proof}
As each node at layer $i>0$ has a neighbor at layer $i-1$, each node will be assigned a gather-label.  

\noindent \textbf{(a, b):}~The values of gather-labels belong to the range $[0,\Delta -1]$ because of the upper bound $\Delta$ on the degrees of nodes
and the rule that, for a child of a node $v$, we choose 
the smallest non-negative integer not assigned earlier to any other neighbor of $v$ at layer $i$.
Property (b) follows from the assigning rules.

\noindent \textbf{(c):}~First, assume that the assignment of gather-labels to the children of $\parent(v)$ is done before the assignment of gather-labels to the children of $\parent(u)$. In this case, if $(u,\parent(v))\in E$ then $u$ is a child of $\parent(v)$ which contradicts the assumption $\parent(v)\neq \parent(u)$.
Now, assume that the assignment of gather-labels to the children of $\parent(u)$ is done before the assignment of gather-labels to the children of $\parent(v)$. In this case, if $(u,\parent(v))\in E$ then none of the children of $\parent(v)$, including $v$, is assigned the gather-label $g_u$ which contradicts the assumption $g_u=g_v$.
%
\end{proof}



Since the label $\mathcal{L}(v)$ of any node is a concatenation of three bits and three representations of integers at most $\Delta$, our labeling scheme has length $O(\log \Delta)$.

\subsubsection{Broadcast and Gathering Algorithms}\label{ss:bandg:bfs}

\paragraph{Broadcast algorithm.}

Our broadcast algorithm \textsc{BroadcastBFS} works in phases. Let $D^*$ be the depth of the tree $T_{\text{BFS}}$. Note that $D^*$ is $\Theta(D)$, where $D$ is the diameter of the graph. In phase $i\in\{0,...,D^*-1\},$ we want to deliver the message $M$ from layer $i$ to layer $i+1$
of $T_{\text{BFS}}$. Each phase consists of {$\Delta+1$} rounds, which gives {$D^*(\Delta+1)$} rounds in total. 
A node $v$ transmits only in round $\layer(v)\cdot\Delta + b_v + 1$.
In the first round of the broadcast algorithm only the root, i.e., the node with $r_v = 1$ transmits. Every other node $v$ registers the round number in which it receives a message for the first time. From that information, it can deduce the value of $\layer(v)$ and, as it knows the value of $\Delta$ from its label, it is able to determine the round of its own transmission of the broadcast message (which is equal to the $b_v$th round of the next phase).

\paragraph{Acknowledged broadcast.}

For the acknowledged version \textsc{AckBrBFS} of broadcast, we first execute the above broadcast algorithm \textsc{BroadcastBFS}. However, instead of sending just the broadcast message $M$, each node $v$ transmits the pair $(0,M)$. When the chosen node $v$ with $\leaf(v)=1$ and $a_v = 1$ receives the message $(0,M)$, it can determine the maximum layer $D^*$ in the tree $T_{\text{BFS}}$ (based on the round number of reception of the message).
After the last round of broadcast, which is the round $D^*(\Delta+1)$, $v$ transmits the pair $(1,D^*)$. Whenever a node $u$ with $a_u = 1$ receives the message $(1,D^*)$, it transmits the same message in the next round. When $(1, D^*)$ reaches the root $r$, node $r$ starts the next execution of \textsc{BroadcastBFS} with the broadcast message equal to $D^*+2D^* (\Delta+1)$, i.e., the number of rounds of the whole execution of \textsc{AckBrBFS}.


\paragraph{Gathering algorithm.}

The gathering algorithm \textsc{GatherBFS} starts by an execution of the acknowledged broadcast algorithm \textsc{AckBrBFS} followed by another execution of the broadcast algorithm \textsc{BroadcastBFS} with the broadcast message $M$ equal to the value of $D^*$, {which is known to the source node at the beginning of this execution}.
After that, we can
assume that each node $v$ knows $\layer(v)$ and $D^*$.
The main part of the gathering algorithm \textsc{GatherBFS}
also consists of $D^*$ phases $i=0,...D^*-1$, each phase lasting $ \Delta$ rounds. If a node is at layer $i$, it listens in phase $D^*-i-1$ and stores all messages received in this phase in its local memory. 
A node $v$ at
layer $D^*-i$ transmits in round $g_v$ of the $i$th phase. The message transmitted by $v$ contains its original message $M_v$ and all messages that $v$ received until the round of its transmission. All the phases 
end in $D^*\Delta$ rounds and, by Lemma~\ref{l:gv:bits}, each node receives the messages transmitted by all its children. In particular, the root receives messages of all nodes of the tree. 

The description of the labeling scheme and the properties of the above algorithms imply the following lemma.
\begin{lemma}\label{l:BFS:broadcast:gather}
	The algorithms \textsc{BroadcastBFS}, \textsc{AckBrBFS} and \textsc{GatherBFS} solve the broadcast problem, the acknowledged broadcast problem and the gathering problem, respectively,
	using a labeling scheme of length $O(\log\Delta)$, and work in time
	 $O(D\Delta)$.
\end{lemma}

\subsection{Topology recognition algorithm}\label{ss:tp:algorithm}
In order to solve the topology recognition problem, we choose an arbitrary node as the root $r$ and assign to each node $v$ the label which is the concatenation of two binary strings: the label $\mathcal{L}(v)$ used by our broadcast-gathering primitive and the binary representation of a natural number $\mathcal{C}(v)$ called the \emph{color} of $v$  which is the color of $v$ in some fixed distance-two vertex coloring of the graph $G$ with the set of colors equal to $[1,\Delta^2]$. Given these labels, we execute the following four-stage algorithm \textsc{TopRec}:

\noindent\textbf{Stage~1.} First, the acknowledged broadcast \textsc{AckBrBFS} is executed with the modification that each node $v$ transmits the concatenation of the consecutive substrings $g_u$ of the labels of nodes of the path from the root $r$ to $v$, including $g_v$. This concatenation will be denoted by ID($v$). 


\noindent\textbf{Stage~2.} Then, in the block of $\Delta^2$ rounds, each node $v$ such that $\mathcal{C}(v)=i\in[1,\Delta^2]$ transmits in the round $i$ of the block, sending the string ID($v$), and each node stores received messages in its local memory.

\noindent\textit{Remark.} If $\Delta^2>n$ then $\log\Delta\ge \frac12\log n$. In this case, in order to reduce the number of rounds of Stage~2 to $\min(\Delta^2,n)$, we extend the labeling by assigning to each node a unique identifier in the range $[1,n]$ of length $O(\log n)=$ $O(\log\Delta)$. Then, Stage~2 consists of $n$ rounds, where the node with the identifier $i$ transmits in round $i$ of the stage.

\noindent\textbf{Stage~3.}
Next, we execute the gathering algorithm \textsc{GatherBFS}, where the message $M_v$ of a node $v$ is equal to the set of IDs received by $v$ in Stage~2, together with its own ID. 
The set of these messages, over all nodes $v$, permits each node to reconstruct the topology of the graph and situate itself in it.

\noindent\textbf{Stage~4.}
Finally, we execute the broadcast algorithm \textsc{BroadcastBFS}, where the message $M$ is equal to the set of messages of all nodes gathered at the root $r$ in Stage~3. 


\begin{theorem}\label{t:tr:algorithm}
	The algorithm \textsc{TopRec} solves the topology recognition problem on every graph, using a labeling scheme of length $O(\log\Delta)$. It works in time $O(D\Delta+\min(n,\Delta^2))$.
	\end{theorem}
\begin{proof}
	As shown in Lemma~\ref{l:BFS:broadcast:gather}, Stages~1, 3 and 4 work in $O(D\Delta)$ rounds.
	Moreover, as the children of each node in $T_{\text{BFS}}$ have different gather-labels, the IDs assigned to nodes in Stage~1 are unique. In Stage~2, lasting $\min(\Delta^2,n)$ rounds, each node $v$ receives a message from each of its neighbors in the graph $G$ and it learns the distinct IDs of these neighbors. Hence the message $M_v$ that each node $v$ sends in Stage 3 contains the information about the neighborhood of $v$ with distinct ID's.Thus, in Stage~3, the root gathers information about all edges in $G$ and it broadcasts this information to all nodes in Stage~4. 
\end{proof}

\section{Conclusion}
We constructed labeling schemes of asymptotically optimal length for size discovery and topology recognition in arbitrary radio networks without collision detection. We also designed algorithms to solve these problems using these optimal schemes. Our results show that, for arbitrary radio networks, the optimal length of labeling schemes permitting topology recognition is exponentially larger than that permitting the easier problem of size discovery.

In the case of size discovery, we showed that our algorithm using the labeling scheme of optimal length $O(\log\log \Delta)$ has also asymptotically optimal time among size discovery algorithms using optimal schemes for this problem. However, for topology recognition, our algorithm using the labeling scheme of optimal length $O(\log \Delta)$  works in time $O\left(D\Delta+\min(\Delta^2,n)\right)$  and we do not know how far this time is from optimal. Hence the main open problem left by our research is the following: What is the time of the fastest topology recognition algorithm using a labeling scheme of optimal length $O(\log \Delta)$?

	

\ifbibtex	
\bibliographystyle{abbrv}
\bibliography{references}
\else

\fi

\end{document}